\documentclass[11pt]{article}

\usepackage{fullpage}

\usepackage{amsmath,amssymb,amsthm,thmtools,thm-restate}
\usepackage{cite}
\usepackage{color}
\usepackage{booktabs} 
\usepackage{mathtools} 

\newcommand{\calC}{\ensuremath{\mathcal{C}}}
\newcommand{\calP}{\ensuremath{\mathcal{P}}}
\newcommand{\E}{\mathop{\mathbb{E}}}
\newcommand{\ent}{\mathop{\text{H}}}
\newcommand{\I}{\mathop{\text{I}}}

\newcommand{\size}[1]{\ensuremath{\left|#1\right|}}
\newcommand{\set}[1]{\left\{ #1 \right\}}
\DeclarePairedDelimiter{\floor}{\lfloor}{\rfloor}
\DeclarePairedDelimiter{\ceil}{\lceil}{\rceil}

\newcommand{\alice}{{\sf Alice}}
\newcommand{\bob}{{\sf Bob}}
\newcommand{\cgst}{CONGEST}
\newcommand{\twos}{Algorithm~${\tt 2S}$}
\newcommand{\twop}{Algorithm~${\tt 2P}$}
\newcommand{\fourap}{Algorithm~${\tt 4AP}$}
\newcommand{\fourp}{Algorithm~${\tt 4P}$}
\newcommand{\eightap}{Algorithm~${\tt 8AP}$}
\newcommand{\pc}{{\sf PART-COMP$_{m,p}$}}

\DeclareMathOperator{\dist}{\delta}

\declaretheorem{theorem}
\declaretheorem[numberlike=theorem]{lemma}
\declaretheorem[numberlike=theorem]{corollary}

\declaretheorem[numberlike=theorem]{conjecture}

\newtheorem*{theorem*}{Theorem}

\begin{document}

\title{Distributed Construction of Purely Additive Spanners%
\thanks{An extended abstract of this work will be presented in DISC 2016.}}
\author{Keren Censor-Hillel\footnotemark[2]
        \and Telikepalli Kavitha\footnotemark[3]
        \and Ami Paz\footnotemark[2]
        \and Amir Yehudayoff\footnotemark[4]}

\renewcommand*{\thefootnote}{\fnsymbol{footnote}}
\footnotetext[2]{Department of Computer Science, Technion, Israel. \texttt{\{ckeren,amipaz\}@cs.technion.ac.il}.
Supported by ISF individual research grant 1696/14.
Part of this work was done while Ami Paz was
visiting TIFR, Mumbai.}
\footnotetext[3]{Tata Institute of Fundamental Research, India. \texttt{kavitha@tcs.tifr.res.in}.}
\footnotetext[4]{Department of Mathematics, Technion, Israel.
\texttt{amir.yehudayoff@gmail.com}.}
\renewcommand*{\thefootnote}{\arabic{footnote}}

\maketitle

\begin{abstract}
This paper studies the complexity of distributed construction of purely additive spanners in the CONGEST model.
We describe algorithms for building such spanners in several cases. Because of the need to simultaneously make decisions at far apart locations, the algorithms use additional mechanisms compared
to their sequential counterparts.

We complement our algorithms with a lower bound on the number of rounds
required for computing pairwise spanners. The standard reductions from set-disjointness and equality seem unsuitable for this task because no specific edge needs to be removed from the graph.
Instead, to obtain our lower bound, we define a new communication
complexity problem that reduces to computing a sparse spanner, and
prove a lower bound on its communication complexity using information theory. This technique significantly extends the current toolbox used for obtaining lower bounds for the CONGEST model, and we believe it may find additional applications.
\end{abstract}

\section{Introduction}
A graph \emph{spanner} is a sparse subgraph that guarantees some bound on how much the original distances are stretched. Graph spanners, introduced in 1989~\cite{PS89, PelegU89a}, are fundamental graph structures which are central for many applications, such as synchronizing distributed networks~\cite{PelegU89a}, information dissemination~\cite{Censor-HillelHKM12}, compact routing schemes~\cite{Chechik13a, PelegU89b,ThorupZ01}, and more.

Due to the importance of spanners, the trade-offs between their possible sparsity and stretch have been the focus of a huge amount of literature. Moreover, finding time-efficient constructions of spanners with optimal guarantees has been a major goal for the distributed computing community, with ingenious algorithms given in many studies (see, e.g.,~\cite{EP04, E05, ElkinZ06, DerbelGPV09, BaswanaKMP10, BS03, DerbelG08, DerbelGP07, DerbelGPV08, DubhashiMPRS05, Pettie10}).
One particular type of spanners are \emph{purely additive spanners}, in which the distances are promised to be stretched by no more than an additive term. However, distributed constructions of such spanners have been scarce, with the only known construction being a $(+2)$-additive spanner construction with $O(n^{3/2}\log{n})$ edges in $O(\sqrt{n}\log{n} + D)$ rounds in a network of size $n$ and diameter $D$ \cite{LenzenP13} (also follows from~\cite{HolzerW12}).

The absence of distributed constructions of purely additive spanners is explicitly brought into light by Pettie~\cite{Pettie10}, and implicitly mentioned in~\cite{DerbelGP07}.

This paper remedies this state of affairs, by providing a study of the complexity of constructing sparse purely additive spanners in the synchronous CONGEST model~\cite{Peleg_2000}, in which each of $n$ nodes can send an $O(\log{n})$-bit message to each of its neighbors in every round. Our contribution is twofold: first, we provide efficient constructions of several spanners with different guarantees, and second, we present new lower bounds for the number of rounds required for such constructions, using tools that are not standard in this context.

\subsection{The Challenge}

A subgraph $H$ of  an undirected unweighted graph $G=(V,E)$ is called a
purely additive spanner with stretch $\beta$ if for every every pair $(u,v) \in V\times V$, we have
$\dist_H(u,v) \le \dist_G(u,v) + \beta$, where $\dist_H(u,v)$ is the $u$-$v$ distance in $H$ and
$\dist_G(u,v)$ is the $u$-$v$ distance in $G$. The goal in spanner problems is to construct a subgraph
$H$ that is as sparse as possible with $\beta$ as small as possible, i.e., we seek a sparse subgraph of $G$
which approximates all distances with a small stretch.

The problem of computing sparse spanners with small stretch $\beta$ is well-studied and we know how to construct sparse purely
additive spanners for $\beta = 2,4,6$. These have sizes $O(n^{3/2})$~\cite{ACIM99}, $\tilde{O}(n^{7/5})$~\cite{Chechik13},
and $O(n^{4/3})$~\cite{BaswanaKMP10}, respectively. In a very recent breakthrough, it was shown that there is {\em no} purely
additive spanner of size at most $n^{4/3}/{2^{O(\sqrt{\log n})}}$~\cite{AbboudB16-2}.

In a bid to get sparser subgraphs than all-pairs spanners with the same stretch, the following relaxation
of {\em pairwise spanners} has attracted recent interest. Here we are given $\calP \subseteq V \times V$:
these are our ``relevant pairs'' and we seek a sparse subgraph which approximates distances between all pairs
in $\cal P$ with a small stretch. That is, for every pair $(u,v) \in \calP$, the graph $H$ should satisfy
$\dist_H(u,v) \le \dist_G(u,v) + \beta$ and for pairs $(u,v)$ outside $\calP$, the value $\dist_H(u,v)$ could be
arbitrarily large. Such a subgraph $H$ is called a $(+\beta)$-pairwise spanner.
We use $\tau(\calP)$ to denote the number of nodes
appearing in $\calP$, i.e.
$\tau(\calP) = \size{\set{u\mid \exists v :
\{u,v\}\in \calP}}$.

The problem of constructing sparse pairwise spanners was first studied by Coppersmith and Elkin~\cite{CE05} who showed
sparse subgraphs where distances for pairs in $\calP$ were {\em exactly  preserved}; these subgraphs were called
{\em pairwise preservers}.
A natural case for $\calP$ is when $\calP = S \times V$, where $S\subseteq V$ is a set of {\em source} nodes --- here we seek for a sparse
subgraph that well-approximates $s$-$v$ distances for all $(s,v) \in S \times V$. Such pairwise spanners are called
{\em sourcewise spanners}. Another natural setting is when  $\calP = S \times S$ and such pairwise spanners are called
{\em subsetwise spanners}.

Purely additive spanners are usually built in three steps:
first, building clusters which contain all high-degree nodes
and adding all the edges of the unclustered nodes;
second, building BFS trees which $(+2)$-approximate
all the paths with many missing edges;
and third, adding more edges to approximate the other paths.

While our constructions follow the general outline of known sequential constructions of pairwise additive spanners~\cite{Kavitha15,KavithaV13}, their techniques cannot be directly implemented in a distributed setting.
In the sequential setting, the \emph{clustering} phase
is implemented by repeatedly choosing a high-degree node
and adding some of its edges to the spanner;
these neighbors are marked and ignored in the rest of the phase.
In the distributed setting,
going over high degree nodes one by one,
creating clusters and updating the degrees is too costly.
Instead, we choose the cluster centers at random,
as done by Thorup and Zwick~\cite{TZ01},
Baswana and Sen~\cite{BS03},
and Chechik~\cite{Chechik13}
(see also Aingworth et al.~\cite{ACIM99}
for an earlier use of randomization for the
a dominating set problem).

Sources for BFS trees are carefully chosen in the
sequential setting by approximately solving a set-cover problem,
in order to cover all paths with many missing edges.
Once again, this cannot be directly implemented in the distributed setting,
as the knowledge of all paths cannot be quickly
gathered in one location,
so we choose the BFS sources randomly~\cite{Chechik13}.
In both the clustering and BFS phases,
the number of edges increases by a multiplicative $\log^c n$ factor,
for $c<1$.

The main challenge left is to choose additional edges
to add to the spanner in order to approximate the remaining paths well.
To this end, we make heavy use of the parallel-BFS technique
of Holzer and Wattenhofer~\cite{HolzerW12},
which allows to construct BFS trees rooted at $s$ different nodes in $O(s+D)$ rounds.
We use this technique to count edges in a path,
to count missing edges in it,
and to choose which edges to add to the spanner.
Yet, interestingly, we are unable to match the guarantee on the number of edges
of more sophisticated algorithms~\cite{BaswanaKMP10,Kavitha15,Woodruff_2010}.
Some of these algorithms use the \emph{value} of a path,
which is roughly the number of pairs of cluster that
get closer if the path is added to the spanner.
We are not able to measure this quantity efficiently in
the distributed setting,
and this is one of the reasons we are unable to introduce
$(+6)$-all-pairs spanner matching the sequential constructions.

\subsection{Our Contribution}

We provide various spanner constructions in the CONGEST model, as summarized in Tables~\ref{tab: results-size} and~\ref{tab: results-time}.

\begin{table}[htbp]
\small
\centering
\begin{tabular*}{\linewidth}{@{}l@{\extracolsep{\fill}}l@{}r@{}l@{}r@{}}
\toprule
Spanner Type       & \multicolumn{2}{@{}c}{Number of Edges --- Distributed}  & \multicolumn{2}{@{}c}{Number of Edges --- Sequential} \\
\midrule
$(+2)$-sourcewise  &
  $O\left(n^{5/4} \size{S}^{1/4} \log^{3/4} n\right)$ &
  (Thm.\ref{thm: 2s}) &
  $O\left(n^{5/4}\size{S}^{1/4} \log^{1/4}n\right)$ & \cite{KavithaV13} \\

$(+2)$-pairwise  &
  $O\left(n \size{\calP}^{1/3} \log^{2/3}n\right)$ &
  (Thm.\ref{thm: 2p}) &
  $O\left(n \size{\calP}^{1/3}\right) $ & \cite{AbboudB16-1} \\

$(+4)$-pairwise  &
  $O\left(n\size{\calP}^{2/7}\log^{6/7}n \right)$ &
  (Thm.\ref{thm: 4p}) &
  $O\left(n\size{\calP}^{2/7}\log^{3/7}n \right)$ &
  \cite{Kavitha15} \\

$(+4)$-all-pairs  &
  $O\left(n^{7/5} \log^{4/5} n\right)$  &
  (Thm.\ref{thm: 4ap}) &
   $O\left(n^{7/5} \log^{1/5} n\right)$ & \cite{Chechik13} \\

$(+8)$-all-pairs  &
  $O\left(n^{15/11}\log^{10/11}{n}\right)$ &
  (Thm.\ref{thm: 8ap}) &
  $O\left(n^{4/3}\right)$ &
  \cite{BaswanaKMP10} \\

$(+2)$-subsetwise  &
  $O\left(n\size{S}^{2/3} \log^{2/3}n\right)$ &
  (Cor.\ref{cor: 24sub}) &
  $O\left(n\size{S}^{1/2}\right)$ & \cite{Pettie09,CGK13} \\

$(+4)$-subsetwise  &
  $O\left(n\size{S}^{4/7}\log^{6/7}n \right)$ &
  (Cor.\ref{cor: 24sub}) &
  $O\left(n\size{S}^{1/2} \right)$ &
  \cite{Pettie09,CGK13} \\

\bottomrule
\end{tabular*}
\caption{The number of edges in our new constructions versus prior, sequential work.
We compare our $(+4)$-subsetwise with a sequential
construction of a $(+2)$-subsetwise spanner,
and our $(+8)$-all-pairs spanner
with a $(+6)$-all-pairs spanner.}
\label{tab: results-size}
\end{table}

\begin{table}[htbp]
\small
\centering
\begin{tabular*}{\linewidth}{@{}l@{\extracolsep{\fill}}l@{}r@{}l@{}r@{}}
\toprule
Spanner Type       & \multicolumn{2}{@{}c}{Number of Rounds}  & \multicolumn{2}{@{}c}{Lower Bounds} \\
\midrule
$(+2)$-sourcewise  &
  $O\left(\size{S} + D\right)$ &
  (Thm.\ref{thm: 2s}) &
  $\min\set{\tilde\Omega\left(\frac{n^{3/8}}{\size{S}^{1/8}}\right), \Omega \left( D \right)}$ &
  \cite{Pettie10} \\

$(+2)$-pairwise   &
  $O\left(\tau(\calP) + D\right)$&
  (Thm.\ref{thm: 2p}) &
  $\min\set{\tilde\Omega\left(\frac{n^{1/2}}{\size{\calP}^{1/6}}\right),
    \Omega \left(D \right)}$ &
  \cite{Pettie10} \\

&&& $\Omega\left( \frac{\size{\calP}}{n\log n} \right)$ &
  (Thm.\ref{thm: 2p lb}) \\

$(+4)$-pairwise  &
  $O\left(\tau(\calP) + D \right)$ &
  (Thm.\ref{thm: 4p}) &
  $\min\set{\tilde\Omega\left(\frac{n^{1/2}}{\size{\calP}^{1/7}}\right), \Omega \left(D \right)}$ &
  \cite{Pettie10} \\

&&& $\Omega\left( \frac{\size{\calP}}{n\log n} \right)$ &
  (Thm.\ref{thm: beta-p lb}) \\

$(+4)$-all-pairs  &
  $O(n^{3/5} \log^{1/5} n + D)$ &
  (Thm.\ref{thm: 4ap}) &
  $\min\set{\tilde\Omega\left(n^{3/10}\right),
  \Omega \left(D \right)}$ &
  \cite{Pettie10} \\

$(+8)$-all-pairs  &
  $O(n^{7/11}\log^{1/11}{n} + D)$ &
  (Thm.\ref{thm: 8ap}) &
  $\min\set{\tilde\Omega\left(n^{7/22}\right), \Omega \left(D \right)}$ &
  \cite{Pettie10} \\

$(+2)$-subsetwise  &
  $O(\size{S} + D)$ &
  (Cor.\ref{cor: 24sub}) &
  $\min\set{\tilde\Omega\left(\frac{n^{1/2}}{\size{S}^{1/3}}\right), \Omega \left(D \right)}$ &
  \cite{Pettie10} \\

$(+4)$-subsetwise  &
  $O(\size{S} + D)$ &
  (Cor.\ref{cor: 24sub}) &
  $\min\set{\tilde\Omega\left(\frac{n^{1/2}}{\size{S}^{2/7}}\right), \Omega \left(D \right)}$ &
  \cite{Pettie10} \\

\bottomrule
\end{tabular*}
\caption{Running time: algorithms versus lower bounds,
for number of edges as in Table~\ref{tab: results-size}.
$\tilde\Omega$ hides polylogarithmic factors.
  }
\label{tab: results-time}
\end{table}

The distributed spanner construction algorithms we present
have three main properties: stretch, number of edges,
and running time.
All three properties hold w.h.p.\ (with high probability).
That is, the algorithm stops in the desired time, with the desired number of edges and the spanner produced has
the desired stretch with probability $1-n^{-c}$,
where $c$ is constant of choice.
However, we can trade the properties and guarantee two
of the three to always hold:
if the spanner is too dense or the stretch is too large,
we can repeat the algorithm;
if the running time exceeds some threshold,
we can stop the execution and output the whole graph to get $0$ stretch,
or output an empty graph to get the desired number of edges.
The edges of the constructed spanner can be counted over a
BFS tree in $G$ within $O(D)$ rounds.
In sourcewise, subsetwise and pairwise spanners,
the stretch is measured by running BFS from the relevant nodes
(nodes of $S$ of appearing in $\calP$) for $O(D)$ rounds
in $G$ and again in $H$;
in all-pairs spanners,
the stretch is measured by measuring the stretch
of the underlying sourcewise or subsetwise spanner.

We complement our algorithms with some lower bounds
for the \cgst{} model.
We show that any
algorithm that constructs an additive $(+2)$-pairwise spanner
with $m$ edges on $p\leq m$ pairs
must have at least $\Omega(p/(n\log n))$ rounds,
as long as $m \leq n^{3/2}$.
For example, a \cgst{} construction of a $(+2)$-pairwise spanner
must take $\tilde \Omega(\sqrt{n})$ rounds.
We also prove lower bounds for $(\alpha,\beta)$-pairwise spanners
(i.e., for which $\dist_H(u,v) \le \alpha\dist_G(u,v) + \beta$).
We show that any
algorithm that constructs an $(\alpha,\beta)$-pairwise spanner with $m$ edges on $p\leq m$ pairs
must have at least $\Omega(p/(n\log n))$ rounds,
as long as $m \leq n^{1+\frac{4}{9\alpha+3\beta-10}}$,
where the constant in the $\Omega$ notation depends
on $\alpha,\beta$.

We believe the difficulty in obtaining this lower bound arises from the fact that standard reductions from set-disjointness and equality are unsuitable for this task.
At a high level, in most standard reductions the problem boils down
to deciding the existence of an edge
(which can represent, e.g., the intersecting element between the inputs);
when constructing spanners, no specific edge needs to be added to the spanner or omitted from it, so the solution is allowed a considerable amount of slack that is not affected by any particular edge alone.

Instead, to obtain our lower bound, we define a new communication complexity problem that reduces to computing a sparse spanner, and prove a lower bound on its communication complexity using information theory. In this new problem, which we call \pc{}, \alice{} has a set $x \subseteq \set{1,\ldots, m}$ of size $\size{x} = p$, and \bob{} has to output a set $y \subseteq \set{1,\ldots, m}$ of size $\size{y} = m/2$ so that $x \cap y = \emptyset$. We show that any protocol that solves \pc{} must convey $\Omega(p)$ bits of information about the set $x$.
This technique significantly extends the current toolbox used for obtaining lower bounds for the CONGEST model. As such, we believe it may find additional applications, especially in obtaining lower bounds for computing in this model.

\paragraph{Roadmap: }
We conclude this section with a further discussion of related work. Section~\ref{sec:preliminaries} contains the definition of the model and some basic routines.
In Section~\ref{sec: algorithms}
we present distributed algorithms for
computing the various types of spanners discussed above.
In Section~\ref{sec:lowerbounds} we present our new lower bounds, and we conclude with a short discussion in Section~\ref{sec:discussion}.

\subsection{Related Work}
Sparse spanners with a small multiplicative stretch
are well-understood: Alth\"ofer et al.~\cite{ADDJS93} in 1993 showed that any weighted graph $G$ on $n$ vertices has a
spanner of size $O(n^{1+1/k})$ with multiplicative stretch $2k-1$, for every integer $k \ge 1$.
Since then, several works~\cite{BS03,DHZ97,EP04,E05,knudsen14,Pettie09,RMZ05,RZ04,TZ06}
have considered the problem of  efficiently constructing sparse spanners with small stretch and have used spanners in the
applications of computing approximate distances and approximate shortest paths efficiently.

For unweighted graphs, one seeks spanners where the stretch is purely additive and as mentioned earlier, an almost tight
bound of $n^{4/3}$ is known for how sparse a purely additive spanner can be. Bollob\'{a}s et al.~\cite{BCE03} were the first
to study a variant of pairwise preservers called {\em distance preservers}, where the set of relevant pairs is
$\calP = \{(u,v): \dist_G(u,v) \ge d\}$, for a given parameter $d$.
Coppersmith and Elkin~\cite{CE05} showed pairwise preservers of size $O(n\sqrt{|\calP|})$
and $O(n+|\calP|\sqrt{n})$ for any $\calP \subseteq V \times V$.
For $\size{P} = \omega\left(n^{3/4}\right)$,
the bound of $O(n\sqrt{|\calP|})$ for pairwise preservers
has very recently been improved to $O(n^{2/3}|\calP|^{2/3} + n|\calP|^{1/3})$ by Bodwin and Williams~\cite{BW16}.

The problem of designing sparse pairwise spanners was first considered by Cygan et al.~\cite{CGK13} who showed a
tradeoff between the additive stretch and size of the spanner. The current sparsest pairwise spanner with
purely additive stretch has size $O(n|\calP|^{1/4})$ and additive stretch 6~\cite{Kavitha15}.
Woodruff~\cite{Woodruff_2010} and Abboud and Bodwin~\cite{AbboudB16-1,AbboudB16-2} showed lower bounds for additive
spanners and pairwise spanners. Parter~\cite{Par14} showed sparse {\em multiplicative} sourcewise
spanners and a lower bound of $\Omega(n|S|^{1/k}/k)$ on the size of a sourcewise spanner with additive stretch
$2(k-1)$, for any integer $k \ge 1$.

Distributed construction of sparse spanners with multiplicative stretch was addressed in several studies~\cite{BaswanaKMP10, BS03, DerbelG08, DerbelGP07, DerbelGPV08, DubhashiMPRS05, Pettie10}.
Constructions of $(\alpha,\beta)$-spanners were addressed in~\cite{BaswanaKMP10, DerbelGPV09, Pettie10}. Towards the goal of obtaining purely additive spanners, for which $\alpha=1$, Elkin and Peleg~\cite{EP04} introduced nearly-additive spanners, for which $\alpha=1+\epsilon$. Additional distributed constructions of nearly-additive spanners are given in~\cite{E05, ElkinZ06, DerbelGPV09, Pettie10}. Finally, somewhat related, are constructions of various spanners in the streaming model, and in dynamic settings, both centralized and distributed~\cite{BaswanaKS12, Baswana08, BaswanaS08, Elkin07, Elkin07a}.

In his seminal paper, Pettie~\cite{Pettie10} presents lower bounds for the number of rounds needed by distributed algorithms in order to construct several families of spanners. Specifically, it is shown that computing an all-pair additive $\beta$-spanner with size $n^{1+\rho}$ in expectation, for a constant $\beta$, requires $\Omega\left(n^{(1-\rho)/2}\right)$ rounds of communication. Because this is an indistinguishability-based lower bound, it holds even for the less restricted LOCAL mode, where message lengths can be unbounded.

The lower bound is obtained by showing an $n$-node graph with diameter $D=\Theta\left(n^{(1-\rho)/2}\right)$ where, roughly speaking, removing \emph{wrong} edges induces a stretch that is too large, and identifying these wrong edges takes $\Omega(D)$ rounds. This gives a lower bound of $\min\set{\Omega\left(n^{(1-\rho)/2}\right),
\Omega \left(D \right)}$ rounds. Examining the construction in detail, it is not hard to show it works for other types of spanners as well: even for a single pair of nodes, or a set $S$ of size $2$,
at least $\Omega(D)$ rounds are necessary in order to avoid removing wrong edges.

\section{Preliminaries}
\label{sec:preliminaries}
\paragraph{The Model: }
The distributed model we assume is the well-known \cgst{} model~\cite{Peleg_2000}.
Such a system consists of a set of $n$ computational units, who exchange messages according to an undirected
\emph{communication graph} $G=(V,E)$, $\size{V} = n$, where nodes represent the computational units
and edges the communication links. Each node has a unique identifier which can be encoded using $O(\log n)$ bits.
The diameter of $G$ is denoted by $D$.

When the computation starts, each node knows its own identifier and the identifiers of its neighbors;
when there is a set $S$ of nodes or a set $\calP$ of node-pairs involved in the computation,
it also knows if it belongs to $S$,
or all the pairs in $\calP$ it belongs to.
The computation proceeds in rounds, where in each round each node sends an $O(\log n)$-bits message to each of its
neighbors, receives a message from each neighbor, and performs a computation. We use the number of rounds as our
complexity measure, while ignoring the local computation time; however, in our algorithms all local computations take polynomial time.
When the computation ends,
each node knows which of its neighbors is also its neighbor
in the new graph $H=(V,E')$ generated.
We do not assume that the global structure of $H$ is known to any of the nodes.

\paragraph{Clustering and BFS: }
The first building block in all of our algorithms is
\emph{clustering}.
A \emph{cluster} $C_i$ around a
\emph{cluster center} $c_i$
is a subset of $\Gamma_G (c_i)$,
the set of neighbors of $c_i$ in $G$
(which does not include $c_i$ itself).
A node belonging to a cluster
is \emph{clustered},
while the other nodes of $G$
are \emph{unclustered}.
We use $\calC$ to denote the set of
cluster centers
and $\hat\calC$ to denote the set of clusters.

In the \emph{clustering} phase of our algorithms
we divide some of the nodes into clusters.
We create a new graph containing all the edges connecting
a clustered node to its cluster center,
and all the edges incident on unclustered nodes.

Another building block is \emph{BFS trees}.
A BFS tree in a graph $G$,
rooted at a node $r$,
consists of shortest paths from $r$
to all other nodes in $G$.
The process of creating a BFS tree, known as BFS search,
is well-known in the sequential setting.
In the distributed setting,
a single BFS tree can be easily constructed
by a techniques called \emph{flooding}
(see, e.g.~\cite[\S 3]{Peleg_2000}),
and a celebrated result of Holzer and Wattenhofer~\cite{HolzerW12}
asserts that multiple BFS trees,
rooted at a set $S$ of nodes,
can be constructed in $O(\size{S} + D)$ rounds.
Here, $D$ denotes the diameter of the graph,
i.e.\ the maximal distance between two nodes.
We use this technique to add BFS trees
to the spanner we construct,
and to measure distances in the original graph.

\section{Building Spanners}
\label{sec: algorithms}

In this section we present distributed algorithms
for building several types of additive spanners.
For each spanner, we first describe a template
for constructing it independently of a computational model
and analyze its stretch and number of edges.
Then, we provide a distributed implementation of the algorithm
in the \cgst{} model and analyze its running time.

In a nutshell, our algorithms have three steps:
first, each node tosses a coin
to decide if it will serve as a cluster center;
second, each cluster center tosses another coin
to decide if it will serve as a root of a BFS tree;
third, add to the current graph edges that are part
of certain short paths.
The parameters of the coins
and the meaning of ``short'' are carefully chosen,
depending on the input to the problem
and the desired stretch.

Proving that the algorithms perform well
is about analyzing the probability of failure.
This analysis uses the graph structure
as well as standard concentration bounds.
In all of our algorithms,
$c$ is a constant that can be chosen
according to the desired exponent of $1/n$
in the failure probability.

\subsection{A $(+2)$-Sourcewise Spanner}
Our first algorithm constructs a $(+2)$-sourcewise spanner.
Given a set $S\subseteq V$,
the algorithm returns a subgraph $H$ of $G$
satisfying
$\delta_H(s,v) \leq \delta_G(s,v) +2$
for all $(s,v) \in S\times V$,
with
guarantees as given in the following theorem.

\begin{theorem}
\label{thm: 2s}
Given a graph $G$ on $n$ nodes
and a set of sources $S$,
a $(+2)$-sourcewise spanner
with $O(n^{5/4} \size{S}^{1/4} \log^{3/4} n)$ edges
can be constructed in $O(\size{S}+D)$ rounds
in the \cgst{} model
w.h.p.
\end{theorem}

This is only a factor $O(\log^{1/2}n)$ more than the number of edges given by the best sequential algorithm known
for this type of spanners~\cite{KavithaV13}.
Lemmas~\ref{lemma: 2s size} and \ref{lemma: 2s stretch} analyze the size and stretch of \twos{} given below.
The number of rounds of its distributed implementation is analyzed in Lemma~\ref{lemma: 2s complexity}, giving Theorem~\ref{thm: 2s}.

\subsubsection{\twos}

Input: a graph $G=(V,E)$; a set of source nodes $S\subseteq V$\\
Output: a subgraph $H$\\
Initialization:
$n=|V|$,
$h=(n\size{S})^{1/4} \log^{3/4} n$, and $H=(V,\emptyset)$

\paragraph{Clustering}
Pick each node as a \emph{cluster center}
w.p.\ $\frac{c \log n}{h}$,
and denote the set of selected nodes by
$\set{ c_1, c_2, \ldots }$.
For each node $v\in V$,
choose a neighbor $c_i$ of $v$ which
is a cluster center,
if such a neighbor exists,
add the edge $(v,c_i)$ to $H$,
and add $v$ to $C_i$;
if none of the neighbors of $v$
is a cluster center,
add to $H$ all the edges $v$ belongs to.

\paragraph{BFS}
Pick each cluster center as a \emph{root} of a BFS tree
w.p.\ $\frac{h^2}{cn\log n}$,
and add to $H$ a BFS tree rooted at each chosen root.

\paragraph{Path Buying}
For each source-cluster pair $(s, C_i) \in S \times \hat\calC$:
build a temporary set of paths,
containing a single, arbitrary shortest path
from $s$ to each $x\in C_i$;
omit from the set all paths with more than
$\frac{2c^2 n\log^2 n}{h^2}$ missing edges
(i.e. edges in $G$ but not in $H$);
if any paths are left,
add to $H$ the shortest among them.

\subsubsection{Analysis of \twos}

We now study the properties
of the spanner $H$ created by the algorithm;
in the next section,
we describe the implementation of the different phases
in the \cgst{} model
and analyze the running time of the algorithm.

\begin{lemma}
\label{lemma: 2s size}
Given a graph $G=(V,E)$ with $|V|=n$
and a set $S\subseteq V$,
\twos{} outputs a graph $H=(V,E')$,
$E'\subseteq E$,
with $|E'| = O(n^{5/4} \size{S}^{1/4} \log^{3/4} n)$
edges w.p. at least $1-O(n^{-c+2})$.
\end{lemma}

\begin{proof}
The algorithm starts with $H=(V,\emptyset)$,
and adds to it only edges from $G$.
We analyze the number of edges added in each phase.

In the first part of the clustering phase,
each node adds to $H$ at most one edge,
connecting it to a single cluster center,
for a total of $O(n)$ edges.
Then, the probability that a node
of degree at least $h$
is left unclustered
is at most $(1-\frac{c\log n}{h})^{h}$
which is $O(n^{-c})$.
A union bound implies that all nodes
of degree at least $h$
are clustered w.p.~$1-O(n^{-c+2})$,
and thus the total number of edges
added to $H$ by unclustered nodes
in the second part of the clustering phase
is $O(nh)$ w.p.~$1-O(n^{-c+2})$.

A node becomes a root in the BFS phase
if it is chosen as a cluster center
and then as a root,
which happens with probability
$ \frac{c \log n}{h} \cdot \frac{h^2}{cn\log n} = \frac{h}{n}$.
Letting $T$ denote the set of trees gives
$\E[\size{T}] = \frac{h}{n}\cdot n = h$,
and a Chernoff bound implies that
$\Pr[\size{T} >4h] \leq \exp(-h)$.
As $h\geq c\log n$
we have $\exp(-h) = O(n^{-c})$,
and the BFS phase adds at most $4h$ trees,
which are $O(nh)$ edges.

Finally,
each of the $n$ nodes is chosen as a cluster center
with probability $\frac{c\log n}{h}$,
so $\E[\size{\calC}] = \frac{cn\log n}{h}$.
A Chernoff bound implies
$\Pr [\size{\calC} > \frac{4cn\log n}{h} ]
\leq \exp(-\frac{cn\log n}{h})$;
as $h<n^{3/4}$, we have
$\exp(-\frac{cn\log n}{h})< \exp(-cn^{1/4}) = o(n^{-c})$.
For each pair in $S\times \hat\calC$,
at most $\frac{2c^2 n\log^2 n}{h^2}$ edges are added
in the path buying phase,
for a total of
$O(\size{S} \cdot \frac{n\log n}{h}\cdot
    \frac{n\log^2 n}{h^2})=
O(\frac{\size{S}n^2 \log^3 n}{h^3})$
edges.

Substituting $h=(n\size{S})^{1/4} \log^{3/4} n$ gives
a total of $O(n^{5/4} \size{S}^{1/4} \log^{3/4} n)$ edges,
as claimed.
\end{proof}

\begin{lemma}
\label{lemma: 2s stretch}
Given a graph $G=(V,E)$ with $|V|=n$
and a set $S\subseteq V$,
the graph $H$ constructed by \twos{}
satisfies $\delta_H(s,v) \leq \delta_G(s,v) +2$
for each pair $(s,v)\in S\times V$
w.p. at least~$1-O\left(n^{-c}\right)$.
\end{lemma}

\begin{proof}
Consider a shortest path $\rho$
between $s\in S$ and $v\in V$ in $G$.

If $\rho$ has more than $\frac{2c^2 n\log^2 n}{h^2}$
missing edges in $H$
after the clustering phase,
then it traverses more than
$\frac{c^2 n\log^2 n}{h^2}$ clusters,
as otherwise there is a shorter path
between $s$ and $v$ in $G$.
The probability that none of the centers of these clusters
is chosen as a root in the BFS phase
is at most
$\left( 1-\frac{h^2}{cn\log n} \right)^{c^2 n\log^2 n/h^2}
=O( \exp(-c\log n) ) = O\left(n^{-c}\right).$
Let $C_i$ be a cluster that $\rho$ traverses,
and let $u$ be a node in $\rho \cap C_i$.
Adding a BFS tree rooted at $c_i$ ensures that
$\delta_H(s,c_i) = \delta_G(s,c_i) \leq \delta_G(s,u)+1$
and similarly
$\delta_H(c_i,v) = \delta_G(c_i,v) \leq \delta_G(u,v)+1$.
By the triangle inequality
$$\delta_H(s,v)
\leq \delta_H(s,c_i) + \delta_H(c_i,v)
\leq \delta_G(s,u) + \delta_G(u,v) +2 $$
which equals
$\delta_G(s,v) +2$
since $u$ is on $\rho$.
This completes the proof
for $\rho$ with many missing edges.

Consider the complementary case, where $\rho$
has at most $\frac{2c^2 n\log^2 n}{h^2}$ missing edges
in $H$
after the clustering phase.
If $\rho$ traverses no clusters,
then it is contained in $H$,
and $\delta_H(s,v) = \delta_G(s,v)$.
Otherwise, if $v$ belongs to some cluster $C_i$,
then there is a node $v'\in C_i$ (possibly $v$ itself)
such that the shortest path between $s$ and $v'$
is added to the graph $H$
in the path buying phase.
The nodes $v$ and $v'$ both belong to the same cluster,
so $\delta_H (v,v') \leq 2$
and the triangle inequality implies
$\delta_H(s,v) \leq \delta_H(s,v') + \delta_H(v',v)
\leq \delta_H(s,v') + 2$, as claimed.
Finally,
consider the case where $\rho$ traverses at least one cluster
and $v$ is unclustered;
let $u$ be the clustered node closest to $v$
on $\rho$.
The sub-path from $v$ to $u$
is contained in $H$,
so $\delta_H(v,u) = \delta_G(v,u)$
and by the previous analysis
$\delta_H(s,u) \leq \delta_G(s,u) +2$;
the triangle inequality implies
$$\delta_H(s,v) \leq
\delta_H(s,u) + \delta_H(u,v)
\leq \delta_G(s,u) +2 + \delta_G(u,v)$$
and since $u$ is on $\rho$,
we have
$\delta_G(s,u) + \delta_G(u,v) = \delta_G(s,v)$
and the claim follows.
\end{proof}

\subsubsection{Implementing \twos{} in the \cgst{} Model}

We now discuss the implementation of \twos{}
in the \cgst{} model.

\begin{lemma}
\label{lemma: 2s complexity}
\twos{} can be implemented in $O(\size{S}+D)$ rounds
in the \cgst{} model,
w.p. at least~$1-o(n^{-c})$.
\end{lemma}

\begin{proof}
We present distributed implementations for each of the phases
in \twos,
and analyze their running time.

\paragraph{Preprocessing}
In order to run the algorithm properly,
we need each node to know the parameter $h$,
which in turn depends on $\size{S}$ and $n$.
These parameters are not given in advance
to all graph nodes,
but they can be gathered along a BFS tree
rooted at a predetermined node,
e.g. the node with minimal identifier,
and then spread to all the nodes
over the same tree.
This is done in $O(D)$ rounds.

\paragraph{Clustering}
The clustering phase is implemented as follows:
first,
each node becomes a cluster center
w.p. $\frac{c\log n}{h}$
and sends a message to all its neighbors;
then,
each node that gets at least one message
joins a cluster of one of its neighbors,
by sending a message to that neighbor
and adding their connecting edge to the graph;
finally,
nodes that are not neighbors of any cluster center
send a message to all their neighbors
and add all their incident edges to the graph.
The round complexity of this phase is constant.

\paragraph{BFS}
Each cluster center becomes a root of a BFS tree
w.p. $\frac{c h^2}{n\log n}$,
which is done without communication.
Then,
all BFS roots run BFS searches in parallel.
The number of BFS trees is $O(h)$ w.p.~$o\left(n^{-c}\right)$,
as seen in the proof of Lemma~\ref{lemma: 2s size},
and this number of BFS searches can be run in parallel
in $O(D+h)$ rounds,
using an algorithm
of Holzer and Wattenhofer~\cite[\S 6.1]{HolzerW12}.
Their algorithm outputs the distances along the BFS trees,
whereas we wish to mark the BFS tree edges
and add them to the graph;
this requires a simple change to the algorithm,
which does not affect its correctness
or asymptotic running time.

\paragraph{Path Buying}
This phase starts with measuring all the distances
between pairs of nodes in $S\times V$,
and the number of missing edges in each shortest path
measured.
To find all distances from a node $s\in S$
to all other nodes,
we run a BFS search from $s$;
moreover, we augment each BFS procedure with a counter
that counts the missing edges in each path
from the root to a node on the BFS tree.
Running BFS searches from all the nodes of $S$
is done in $O(\size{S}+D)$ rounds, as before,
and adding a counter does not change the
time complexity.
When a node $v\in V$ receives a message
of a BFS initiated by some $s\in S$,
it learns its distance from $s$
and the number of missing edges on one shortest path
from $s$ to $v$,
which lies within the BFS tree;
we refer to this path as \emph{the} shortest path
from $s$ to $v$.

After all the BFS searches complete,
each clustered node $x$ sends to its cluster center
the distance from each $s\in S$ to $x$,
and the number of missing edges
on the corresponding path.
This sub-phase takes $O(\size{S})$ rounds to complete.

Each cluster center $c_i$ now knows,
for each $s\in S$,
the length of the shortest path from $s$
to each $x\in C_i$,
and the number of missing edges in each such path;
it then locally chooses the shortest
among all paths with at most
$\frac{2c^2 n\log^2 n}{h^2}$ missing edges.
Finally,
for each chosen $(x,s)$ path,
$c_i$ sends a message to $x$
containing the identifier of $s$.
All BFS searches are now executed backwards,
by sending all the messages in opposite direction and order;
when $x$ runs backwards the BFS search initiated by $s$,
it marks the message to his parent
with a ``buy'' bit,
which is passed  up the tree
and makes each of its receivers
add the appropriate edge to the graph.
This sub-phase requires $O(\size{S}+D)$ rounds as well.

In total, the running time of the algorithm
is $O(h+\size{S}+D)$, w.p. at least~$1-o(n^{-c})$,
which completes the proof for the case
$h=O(\size{S})$.
In the case $h=\Omega(\size{S})$,
we can replace the algorithm
by a simpler algorithm
that returns the union of BFS trees rooted
at all nodes of $S$.
This creates a graph that exactly preserves all distances
among pairs $(s,v)\in S\times V$,
and takes $O(\size{S} +D)$ rounds to complete.
The number of edges in the created spanner
is $O(n\size{S})$,
and the assumption $h=\Omega(\size{S})$
implies $O(n\size{S}) = O(nh)$,
as desired.
\end{proof}

\subsection{A $(+4)$-All-Pairs Spanner}

Recall that a subgraph $H$ of $G$ is a
$(+4)$-\emph{all-pairs} spanner if
$\delta_H(u,v) \leq \delta_G(u,v) +4$
for all pairs $(u,v)\in V\times V$.
We present an algorithm, based on \twos{},
which builds a $(+4)$-all-pairs spanner
and has the properties guaranteed by the following theorem.

\begin{theorem}
\label{thm: 4ap}
Given a graph $G$ on $n$ nodes,
a $(+4)$-all-pairs spanner
with $O\left(n^{7/5} \log^{4/5} n\right)$ edges
can be constructed in $O(n^{3/5} \log^{1/5} n + D)$ rounds
in the \cgst{} model w.h.p.
\end{theorem}

The main idea is that cluster centers are now sources for a $(+2)$-sourcewise spanner, which, as we show, promises a $(+4)$-stretch to all pairs.
Lemmas~\ref{lemma: 4ap size} and \ref{lemma: 4ap stretch} analyze the size and the stretch of \fourap{} below. Lemma~\ref{lemma: 4ap complexity} analyzes the running time of its distributed implementation, completing the proof of Theorem~\ref{thm: 4ap}.

\subsubsection{\fourap}
Input: $G=(V,E)$\\
Output: a subgraph $H$\\
Initialization:
$n=|V|$, $h= n^{2/5} \log^{4/5} n$, and $H=(V,\emptyset)$

\paragraph{Clustering}
Run clustering as in \twos.

\paragraph{BFS and Path Buying}
Run the BFS and path buying phases from \twos,
with cluster centers as sources, i.e.\ $S = \calC$.

\subsubsection{Analysis of \fourap}

\begin{lemma}
\label{lemma: 4ap size}
Given a graph $G=(V,E)$ with $|V|=n$,
\fourap{} outputs a graph $H=(V,E')$,
$E'\subseteq E$,
with $|E'| = O\left(n^{7/5} \log^{4/5} n\right)$ edges
w.p. at least\ $1-O(n^{-c+2})$.
\end{lemma}

\begin{proof}
The lemma follows from the proof of Lemma~\ref{lemma: 2s size}:
in \fourap, $S$ is the set of all cluster centers,
whose amount is $\size{\calC}$,
and by the proof $\size{\calC} \leq \frac{4cn\log n}{h}$ w.h.p.
Substituting $\size{S} =
O\left(\frac{n\log n}{h}\right)$
and $h = n^{2/5} \log^{4/5} n$
in Lemma~\ref{lemma: 2s size},
we get that the graph created by \fourap{}
contains $O(n^{7/5} \log^{4/5} n)$ edges
w.p. at least\ $1-O(n^{-c+2})$.
\end{proof}

\begin{lemma}
\label{lemma: 4ap stretch}
Given a graph $G=(V,E)$ with $|V|=n$,
\fourap{} outputs a graph $H$
satisfying $\delta_H(u,v) \leq \delta_G(u,v) +4$
for each pair of vertices $(u,v)\in V\times V$
w.p. at least\ $1-O(n^{-c})$.
\end{lemma}

\begin{proof}
Let $(u,v)\in V\times V$ be an arbitrary pair of nodes,
and set a shortest path $\rho$ in $G$ between them.

If $\rho$ is not incident on any clustered node,
then all its nodes are unclustered
and all its edges are present in $H$.
Otherwise,
let $x$ be the first clustered node on $\rho$,
when traversing it from $u$ to $v$,
and let $C_i$ be the cluster containing $x$.
The sub-path of $\rho$ from $u$ to $x$ exists in $H$,
as all nodes on this sub-path except for $x$
are unclustered;
the distance from $c_i$ to $v$ satisfies
$\delta_H(c_i,v) \leq \delta_G(c_i, v)+2$,
as the stretch of $(c_i,v)$ in $H$ is at most 2
by Lemma~\ref{lemma: 2s stretch},
w.p. at least\ $1-O(n^{-c})$.
The triangle inequality completes the proof:
\begin{equation*}
\begin{split}
\delta_H(u,v)
& \leq \delta_H(u,x) + \delta_H(x,c_i) + \delta_H(c_i,v) \\
& \leq \delta_G(u,x) + 1 + \delta_G(c_i,v) +2 \\
& \leq \delta_G(u,x) + \delta_G(c_i,x) + \delta_G(x,v) + 3
  =    \delta_G(u,v) +4.
\end{split}
\end{equation*}
\end{proof}

\subsubsection{Implementing \fourap}
Running \fourap{} is
done by executing \twos{}
with a specific set $S$;
thus, their running times are identical,
as stated in the next lemma.

\begin{lemma}
\label{lemma: 4ap complexity}
\fourap{} can be implemented
in the \cgst{} model
in $O(n^{3/5} \log^{1/5} n + D)$ rounds
w.p. at least~$1-o(n^{-c})$.
\end{lemma}

\begin{proof}
The lemma follows from the proof of Lemmas~\ref{lemma: 2s size} and~\ref{lemma: 2s complexity}:
in \fourap{}, $S$ is the set of cluster centers,
whose amount is $\size{\calC}$,
and by the proof of
lemma~\ref{lemma: 2s size},
$\size{\calC} \leq \frac{4cn\log n}{h}$ w.h.p.
Substituting $\size{S} = O(\frac{n\log n}{h})$
and $h = n^{2/5} \log^{4/5} n$
in Lemma~\ref{lemma: 2s complexity},
we get that the algorithm completes
in $O(n^{3/5} \log^{1/5} n + D)$ rounds, with the claimed probability.
\end{proof}

\subsection{A $(+2)$-Pairwise Spanner}
Recall that a $(+2)$-pairwise spanner,
for a set of pairs $\calP\subseteq V\times V$,
is subgraph $H$ of $G$
satisfying
$\delta_H(u,v) \leq \delta_G(u,v) +2$
for all pairs $(u,v) \in \calP$.
Recall that $\tau(\calP)$ denotes the number of nodes appearing in $\calP$,
i.e.
$\tau(\calP) = \size{\set{u\mid \exists v :
\{u,v\}\in \calP}}$.

We present a distributed algorithm,
\twop{},
which returns a
$(+2)$-pairwise spanner with the properties described
in the following theorem.

\begin{theorem}
\label{thm: 2p}
Given a graph $G$ on $n$ nodes
and a set $\calP$ of pairs of nodes in $G$,
a $(+2)$-pairwise spanner
with $O(n \size{\calP}^{1/3} \log^{2/3}n)$ edges
can be constructed in
$O(\tau(\calP) +D)$ rounds
in the \cgst{} model w.h.p.
\end{theorem}

If $\tau(\calP) < 2c^2 \size{\calP}^{1/3}\log^{2/3} n$,
achieving the desired spanner is simple:
for each $u$ appearing in $\calP$,
add to $H$ a BFS from $u$.
The number of edges is
$O\left(n \tau(\calP)\right)
=O\left(n\size{\calP}^{1/3}\log^{2/3}n\right)$,
the stretch is $0$ for all pairs in $\calP$,
and the running time is $O(\tau(\calP) + D)$,
as desired.
Otherwise, Lemmas~\ref{lemma: 2p size} and \ref{lemma: 2p stretch} prove the claimed size and stretch of \twop{} below.
Lemma~\ref{lemma: 2p complexity} proves the running time of its distributed implementation,
giving Theorem~\ref{thm: 2p}.

\subsubsection{\twop}

Input: $G=(V,E)$, $\calP\subseteq V\times V$\\
Output: a subgraph $H$\\
Initialization:
$n=|V|$, $h= \size{\calP}^{1/3} \log^{2/3}n$, and $H=(V,\emptyset)$

\paragraph{Clustering and BFS}
Run clustering and add BFS trees from
selected cluster centers, as in \twos.

\paragraph{Path Buying}
For each pair $(u,v) \in \calP$,
if the shortest path between $u$ and $v$ in $G$
has at most $\frac{2c^2 n\log^2 n}{h^2}$ missing edges in $H$,
add it to $H$.

\subsubsection{Analysis of \twop}

\begin{lemma}
\label{lemma: 2p size}
Given a graph $G=(V,E)$ with $|V|=n$
and a set $\calP\subseteq V\times V$,
\twop{} outputs a graph $H=(V,E')$,
$E'\subseteq E$,
with $|E'| =O(n \size{\calP}^{1/3} \log^{2/3}n)$ edges,
w.p. at least~$1-O\left(n^{-c+2}\right)$.
\end{lemma}

\begin{proof}
The lemma follows from the proof of
Lemma~\ref{lemma: 2s size}:
the clustering and BFS phases
add $O(nh)$ edges to the graph
w.p. at least~$1-O\left(n^{-c+2}\right)$,
as long as $c\log n \leq h \leq n^{3/4}$.
The first inequality comes from the comment after
the statement of Theorem~\ref{thm: 2p}, the fact that
$\tau(\calP) \leq 2\size{\calP}$ and the choice of $h$,
and the second inequality is immediate.

In the path buying phase,
at most $\frac{2c^2 n\log^2 n}{h^2}$ edges are added
for each pair in $\calP$,
for a total of
$O\left(\frac{\size{\calP} n\log^2 n}{h^2}\right)$ edges.
Substituting $h= \size{\calP}^{1/3} \log^{2/3}n$,
we get a total of $O\left(n \size{\calP}^{1/3} \log^{2/3}n\right)$
edges in $H$.
\end{proof}

\begin{lemma}
\label{lemma: 2p stretch}
Given a graph $G=(V,E)$ with $|V|=n$,
\twop{} outputs a graph $H$
satisfying $\delta_H(u,v) \leq \delta_G(u,v) +2$
for each pair of vertices $(u,v)\in \calP$,
w.p. at least~$1-O\left(n^{-c}\right)$.
\end{lemma}

\begin{proof}
Let $(u,v)\in \calP$ be an arbitrary pair of nodes,
and fix a shortest path $\rho$ in $G$ between them.

If $\rho$ has at most
$\frac{2c^2 n\log^2 n}{h^2}$ missing edges
in $H$ before the path buying phase,
it is added to $H$,
and $\delta_G(u,v) = \delta_H(u,v)$.
Otherwise,
$\rho$ has more than $\frac{2c^2 n\log^2 n}{h^2}$ missing edges
before the BFS phase,
so it traverse at least $\frac{c^2 n\log^2 n}{h^2}$ clusters.
As in the proof of Lemma~\ref{lemma: 2s stretch},
at least one of the corresponding cluster centers
is chosen as a root of a BFS tree w.p. at least~$1-O\left(n^{-c}\right)$,
and $\delta_H(u,v) \leq \delta_G(u,v) +2$,
as claimed.
\end{proof}

\subsubsection{Implementing \twop}
\begin{lemma}
\label{lemma: 2p complexity}
\twop{} can be implemented in
$O(\tau(\calP) + \size{\calP}^{1/3} \log^{2/3}n + D)$ rounds
in the CONGEST model
w.p. at least $1-o\left(n^{-c}\right)$.
\end{lemma}

\begin{proof}
We can implement the clustering and path buying phases in
$O(h+D)= O(\tau(\calP) +D)$ rounds
with success probability~$1-o\left(n^{-c}\right)$,
as seen in the proof of Lemma~\ref{lemma: 2s complexity}.
In order to count missing edges in paths,
we run a BFS search in $G$ from each node appearing in $\calP$.
Then,
the BFS search is run backwards,
and is used to add the ``cheap'' paths:
for a pair $(u,v)$ in $\calP$,
if the BFS from $v$ arrives at $u$
traversing at most $\frac{2c^2 n\log^2 n}{h^2}$ missing edges,
then $u$ sends back a ``buy'' message up the tree,
and the path is added.
We may end up adding two shortest path for a pair
$(u,v)\in \calP$,
but this does not affect the asymptotic number of edges
or the time complexity.
This phase is implemented in $O(\tau(\calP) + D)$ rounds,
by running the $\tau(\calP)$ BFS searches in parallel.
\end{proof}

\subsection{A $(+4)$-Pairwise Spanner}

We present an algorithm for constructing
a $(+4)$-pairwise spanner,
with the parameters described by the following theorem.

\begin{theorem}
\label{thm: 4p}
Given a graph $G$ on $n$ nodes
and a set $\calP$ of pairs,
a $(+4)$-pairwise spanner
with $O\left(n\size{\calP}^{2/7}\log^{6/7}n \right)$ edges
can be constructed in $O(\tau(\calP)+D)$ rounds
in the \cgst{} model w.h.p.
\end{theorem}

If $\size{\calP} < \log^4 n$,
the $(+2)$-pairwise spanner from Theorem~\ref{thm: 2p}
is sparser than the one promised by Theorem~\ref{thm: 4p},
and can be constructed in the same running time.
Otherwise, Lemmas~\ref{lemma: 4p size} and \ref{lemma: 4p stretch} show the claimed size and stretch of \fourp{} below, which together with Lemma~\ref{lemma: 4p complexity}, which analyzes the running time of its distributed implementation, proves Theorem~\ref{thm: 4p}.

\subsubsection{\fourp}

Input: a graph $G=(V,E)$; a set of pairs $\calP\subseteq V\times V$\\
Output: a subgraph $H$\\
Initialization:
$n=|V|$,
$h=\size{\calP}^{2/7}\log^{6/7}n$,
$\ell=\frac{n\log^3 n}{h^{5/2}}$ and $H=(V,\emptyset)$

\paragraph{Clustering and BFS}
Run clustering and add BFS trees from selected
cluster centers, as in \twos.

\paragraph{Prefix-Suffix Buying}
For each pair $(u, v) \in \calP$,
let $\rho$ be a shortest path from $u$ to $v$.
Add to $H$
the first $\ell$ missing edges
and the last $\ell$ missing edges in $\rho$.

\paragraph{Choosing Cluster Centers}
Construct a set $A$ of cluster centers by adding to it
each cluster center independently w.p.
$\frac{16 c\log n}{\ell}$.

\paragraph{Path Buying}
For each pair $(c_1,c_2)\in A\times A$:
fix a set of paths containing
a single shortest path from $c_1$ to each $x\in C_2$;
omit all paths with more than $\frac{2c^2 n\log^2 n}{h^2}$
missing edges in $H$;
if any paths are left,
add to $H$ the shortest among them.

\subsubsection{Analysis of \fourp}

\begin{lemma}
\label{lemma: 4p size}
Given a graph $G=(V,E)$ with $\size{V}=n$
and a set $\calP\subseteq V\times V$,
\fourp{} outputs a graph $H=(V,E')$,
$E'\subseteq E$,
with $|E'| =O\left(n\size{\calP}^{2/7}\log^{6/7}n \right)$
edges w.p. at least $1-O\left(n^{-c+2}\right)$.
\end{lemma}

\begin{proof}
The clustering and BFS phases
add $O(nh)$ edges to the graph
w.p. at least $1-O\left(n^{-c+2}\right)$,
as seen in the proof of Lemma~\ref{lemma: 2s size},
as long as $c\log n <h <n^{3/4}$.
The first inequality comes from the discussion
below the statement of Theorem~\ref{thm: 4p},
and the second is immediate.

In the prefix-suffix buying phase,
at most $O(\ell)$ edges are bought for each pair in $\calP$,
for a total of
$O\left(\size{\calP} \cdot \frac{n\log^3n}{h^{5/2}}\right)
= O\left(n\size{\calP}^{2/7}\log^{6/7}n \right)$ edges.

Finally, in the path buying phase
we add to $H$ at most $\size{A}^2$ paths,
with $O\left(\frac{n\log^2n}{h^2} \right)$
missing edges in each.
Each node is chosen to be a cluster center w.p.
$\frac{c\log n}{h}$,
and then to enter $A$ w.p.\ $\frac{16 c\log n}{\ell}$,
so $\E[\size{A}]
= n \cdot \frac{c\log n}{h} \cdot \frac{16 c\log n}{\ell}
= \frac{16 c^2 n \log^2n}{h\ell}
= 16 c^2 \size{\calP}^{3/7} \log^{2/7} n$.
A Chernoff bound implies
$\Pr\left[
\size{A} > 64 c^2 \size{\calP}^{3/7} \log^{2/7}n \right]
\leq \exp\left(-16 c^2 \size{\calP}^{3/7} \log^{2/7} n \right)
= O(n^{-c})$;
the last equality holds under the assumption
$\size{\calP}\geq \log n$,
as discussed below the statement of Theorem~\ref{thm: 2p}.
Hence,
the number of edges added in the path buying step is
$O\left(\left(\size{\calP}^{3/7} \log^{2/7} n \right)^2 \cdot
 \frac{n\log^2n}{h^2} \right)
=O(\left(n\size{\calP}^{2/7}\log^{6/7}n \right)$ edges.
In total, $H$ has
$O(\left(n\size{\calP}^{2/7}\log^{6/7}n \right)$ edges
w.p. at least $1-O\left(n^{-c+2}\right)$.
\end{proof}

\begin{lemma}
\label{lemma: 4p stretch}
Given a graph $G=(V,E)$ with $|V|=n$,
\fourp{} outputs a graph $H$
satisfying $\delta_H(u,v) \leq \delta_G(u,v) +4$
for each pair of vertices $(u,v)\in \calP$
w.p. at least $1-O\left(n^{-c+2}\right)$.
\end{lemma}

\begin{proof}
Let $(u,v)\in \calP$ be an arbitrary pair of nodes,
and let $\rho$ be an arbitrary shortest path
from $u$ to $v$.
If $\rho$ has at most $2\ell$ missing edges in $H$
after the clustering phase,
it is added to $H$ in the prefix-suffix buying phase
and $\delta_G(u,v) = \delta_H(u,v)$.

Otherwise, the prefix of $\rho$ with $\ell$ missing edges
is incident on at least $\ell/2$ clusters.
Each cluster center is added to $A$ independently
w.p.\ $\frac{16 c\log n}{\ell}$
so the expected number of clusters in $A$ which are also
incident on the prefix is $8c\log n$,
and a Chernoff bound implies that the probability
that less than $4c\log n$ of the centers of these
clusters are chosen to $A$ is at most
$O(\exp(-c\log n)) = O\left(n^{-c}\right)$.
The same argument shows that the suffix of $\rho$
is incident on a cluster in $A$,
and a union bound implies that all prefixes and suffixes
are incident on clusters in $A$ w.p.\ at least
$1- O\left(n^{-c+2}\right)$.

Let $c_1$ be a center of a cluster in $A$ which is
incident on the prefix of $\rho$,
and $c_2$ a center of a cluster incident
on the suffix of $\rho$.
Let $u'$ and $v'$ be nodes in $\rho \cap C_1$ and
$\rho \cap C_2$ respectively,
and let $\sigma$ be a path between $c_1$ and $v'$ in $G$.

If the number of edges of $\sigma$ missing in $H$
after the clustering phase is more than
$\frac{2c^2 n\log^2 n}{h^2}$,
then $\sigma$ is incident on at least
$\frac{c^2 n\log^2 n}{h^2}$ clusters.
In this case, a cluster incident on $\sigma$ is a source
of a BFS tree w.p. at least $1-O\left(n^{-c}\right)$,
as seen in the proof of Lemma~\ref{lemma: 2s stretch}.
Let $C_i$ be such a cluster,
then after adding the BFS trees it holds that
$\delta_H(u',v') \leq 1+ \delta_H(c_1,v')
\leq  1+ \delta_H(c_1,c_i) + \delta_H(c_i,v')
\leq \delta_G(c_1, v') +3
\leq \delta_G(u',v') +4$,
which implies
$\delta_H(u,v) \leq
\delta_H(u,u')+ \delta_H(u',v')+ \delta_H(v',v)
\leq
\delta_G(u,u')+ \delta_G(u',v')+ 4+  \delta_G(v',v)
=\delta_G(u,v)+4$.

If $\sigma$ has less than $\frac{2c^2 n\log^2 n}{h^2}$
missing edges
then a path between $c_1$ and some $x\in C_2$
is added to $H$ in the path buying phase,
satisfying $\delta_H(c_1,x) \leq \delta_G(c_1,v')
\leq \delta_G(u',v')+1$.
Hence,
$\delta_H(u,v) \leq
\delta_H(u,u')+ \delta_H(u',c_1)+ \delta_H(c_1,x)+
        \delta_H(x,v') + \delta_H(v',v)
\leq \delta_G(u,u')+ 1+ \delta_G(u',v')+1+ 2+ \delta_G(v',v) =\delta_G(u,v)+4$,
as required.
\end{proof}

\subsubsection{Implementing \fourp}
\begin{lemma}
\label{lemma: 4p complexity}
\fourp{} can be implemented in
$O\left(\tau(\calP) + D\right)$ rounds in the \cgst{} model
w.p. at least $1-o\left(n^{-c}\right)$.
\end{lemma}

\begin{proof}
We can implement the clustering and path buying phases in
$O(h+D)$ rounds
with success probability~$1-o\left(n^{-c}\right)$,
as seen in the proof of Lemma~\ref{lemma: 2s complexity}.

The prefix-suffix buying phase is implemented as follows:
run a BFS from each $u$ appearing in $\calP$,
counting missing edges on each path.
Roll back the BFS, and when passing a node $v$
such that $(u,v)\in\calP$,
it sends a ``buy-suffix'' message with a counter
initiated to $\ell$;
when a node receives such a counter
it adds the edge to the parent to $H$ and
decreases the counter by $1$,
unless the edge is already on $H$;
if the counter reaches another node in $v'$
satisfying $(u,v')\in\calP$,
it is set to $\ell$ again.
When the counter is $0$, no more edges are bought
but the counter is passed up the tree,
until it arrives at a node whose count of missing edges
in the original BFS was $\ell$.
This node replaces it with a simple ``buy'' message,
adds the edge to its parent to $H$ and sends it to the parent;
each receiver of the ``buy'' message does the same,
all the way to the tree root.
This phase takes $O(\tau(\calP) +D)$ rounds.

The choice of $A$ is simple, requiring no communication.

The path buying phase is implemented similarly to
its implementation in \twos{},
in the proof of Lemma~\ref{lemma: 2s complexity}.
Measuring distances and counting missing edges
on a path from each $c_1\in A$ to each other node
is done by running a BFS from each cluster center with
the appropriate counters;
each clustered node reports its cluster's center
the above parameters in $O(\size{A})$ rounds;
each $c_2\in A$ then chooses which paths to buy,
and reports the appropriate node in its cluster;
running the BFS searches backwards,
each node may send ``buy'' messages up the trees,
as decided by $c_2$.
This phase takes $O(\size{A} +D)$ rounds.

The proof of Lemma~\ref{lemma: 4p size} implies
$\size{A}=O\left(\size{\calP}^{3/7}\log^{2/7}n\right)$
w.p. at least $1-O\left(n^{-c}\right)$.
This, together with the assumption
$\size{\calP}\geq \log^4(n)$ and with the choice of $h$,
imply $O\left(h+\size{A}+\tau(\calP) +D\right)
=O\left(\tau(\calP) +D\right)$;
hence, the above implementation takes
$O\left(\tau(\calP) +D\right)$ rounds with the same probability.
\end{proof}

\subsection{Subsetwise Spanners}
Recall that a $(+\beta)$-subsetwise spanner
for a set $S\subseteq V$ is a subgraph $H$ of $G$ satisfying
$\delta_H(u,v) \leq \delta_G(u,v) +2$
for all pairs $(u,v) \in S\times S$.
We show how to build such
spanners for $\beta=2$ and $\beta=4$,
with $O(n\size{S}^{2/3} \log^{2/3}n)$ edges
and $O(n\size{S}^{4/7} \log^{6/7}n)$
edges respectively,
in $O(\size{S}+D)$ rounds w.h.p.

The algorithms follow immediately from \twop{}
and \fourp{}:
set $\calP = S\times S$ and run \twop{} or \fourp{}.
The number of edges follows from the fact
$\size{\calP} =\size{S}^2$,
and the running time from
$\tau(\calP) = \size{S}$.

Finally, note that in the case
$\size{S}>n^{3/5}\log^{1/5}n$,
\twos{} gives a sparser spanner than \twop{}
in the same running time,
and with a stretch of $+2$
for all $S\times V$ pairs.
Similarly, when $\size{S}>n^{7/10}\log^{-1/10}n$,
\fourap{} gives a sparser spanner than \fourp{}
in a shorter running time,
with a stretch of $+4$ on all pairs of nodes
in the graph.

\begin{corollary}
\label{cor: 24sub}
Given a graph $G$ on $n$ nodes
and a set $S$ of nodes,
a $(+2)$-subsetwise spanner
with $O\left(\min \set{
n\size{S}^{2/3} \log^{2/3}n,
n^{5/4} \size{S}^{1/4} \log^{3/4} n
}\right)$ edges
and a $(+4)$-subsetwise spanner
with $O\left(\min \set{
n\size{S}^{4/7} \log^{6/7}n,
n^{7/5} \log^{4/5} n
}\right)$ edges,
can both be constructed in at most  $O(\size{S}+D)$ rounds
in the \cgst{} model w.h.p.
\end{corollary}

\subsection{A $(+8)$-All-Pairs Spanner}

Recall that a subgraph $H$ of $G$ is a
$(+8)$-\emph{all-pairs} spanner if
$\delta_H(u,v) \leq \delta_G(u,v) +8$
for all pairs $(u,v)\in V\times V$.
We present an algorithm, based on \fourp{},
which builds a $(+8)$-all-pairs spanner
and has the properties guaranteed by the next theorem.

\begin{theorem}
\label{thm: 8ap}
Given a graph $G$ on $n$ nodes,
a $(+8)$-all-pairs spanner
with $O\left(n^{15/11} \log^{10/11}n\right)$ edges
can be constructed in $O(n^{7/11} \log^{1/11} n + D)$  rounds
in the \cgst{} model w.h.p.
\end{theorem}

Lemmas~\ref{lemma: 8ap size} and \ref{lemma: 8ap stretch} provide the required size and stretch of \eightap{} below, while Lemma~\ref{lemma: 8ap complexity} gives the running time f its distributed implementation, proving Theorem~\ref{thm: 8ap}.

\subsubsection{\eightap}
Input: $G=(V,E)$\\
Output: a subgraph $H$\\
Initialization:
$n=|V|$, $h= n^{4/11} \log^{10/11}n$, and $H=(V,\emptyset)$

\paragraph{Clustering}
Run clustering as in \fourp{}.

\paragraph{Rest of \fourp{}}
Run the rest of \fourp{}
on all pairs of cluster centers,
i.e.\ $\calP = \calC\times\calC$.

\subsubsection{Analysis of \eightap}

\begin{lemma}
\label{lemma: 8ap size}
Given a graph $G=(V,E)$ with $|V|=n$,
\eightap{} outputs a graph $H=(V,E')$,
$E'\subseteq E$,
with $|E'| = O\left(n^{15/11} \log^{10/11}n\right)$ edges
w.p. at least $1-O(n^{-c+2})$.
\end{lemma}

\begin{proof}
We follow the outline of previous proofs.
By the proof of Lemma~\ref{lemma: 2s size},
the BFS phase adds $O(nh)$ edges, and
$\size{\calC} \leq \frac{4cn\log n}{h}$,
w.p. at least $1-O\left(n^{-c+2} \right)$.
In \eightap{}, $\size{\calP} = \size{\calC}^2$,
and substituting
$\size{\calP} = O\left(\frac{n^2\log^2 n}{h^2}\right)$
and $h= n^{4/11} \log^{10/11}n$
in Lemma~\ref{lemma: 4p size}
gives that the graph created by \eightap{}
contains $O(n^{15/11} \log^{10/11}n)$ edges
w.p. at least $1-O\left(n^{-c+2} \right)$.
\end{proof}

\begin{lemma}
\label{lemma: 8ap stretch}
Given a graph $G=(V,E)$ with $|V|=n$,
\eightap{} outputs a graph $H$
satisfying $\delta_H(u,v) \leq \delta_G(u,v) +8$
for each pair of vertices $(u,v)\in V\times V$
w.p. at least $1-O(n^{-c})$.
\end{lemma}

\begin{proof}
Let $(u,v)\in V\times V$ be an arbitrary pair of nodes,
and set a shortest path $\rho$ in $G$ between them.

If $\rho$ is not incident on any clustered node,
then all its nodes are unclustered
and all its edges are present in $H$.
Otherwise,
let $x$ be the first clustered node on $\rho$,
when traversing it from $u$ to $v$,
and let $y$ the last clustered node on $\rho$.
Let $C_1$ be the cluster containing $x$,
and $C_2$ the cluster containing $y$.
The sub-paths of $\rho$ from $u$ to $x$ and from $y$ to $v$
exist in $H$,
as all the nodes on these sub-path except for $x$ and $y$
are unclustered.
By Lemma~\ref{lemma: 4p stretch},
$\delta_H(c_1, c_2) \leq \delta_G(c_1, c_2) +4$
as $c_1,c_2\in \calC = S$;
moreover, $\delta_G(c_1, c_2) \leq  \delta_G(x,y) +2$
by the triangle inequality.
Finally,
\begin{equation*}
\begin{split}
\delta_H(u,v)
& \leq \delta_H(u,x) + \delta_H(x,c_1) + \delta_H(c_1,c_2) + \delta_H(c_2,y) + \delta_H(y,v)\\
& \leq \delta_G(u,x) + 1 \delta_G(x,y) + 6 + \delta_G(y,v)+1
  =    \delta_G(u,v) +8
\end{split}
\end{equation*}
as desired.
\end{proof}

\subsubsection{Implementing \eightap{}}
Running \eightap{} is
done by executing \fourp{}
with a specific set $\calP$;
thus, their running times are identical,
as stated in the next lemma.

\begin{lemma}
\label{lemma: 8ap complexity}
\eightap{} can be implemented in the \cgst{} model
in $O(n^{7/11} \log^{1/11} n + D)$ rounds
w.p. at least $1-o(n^{-c})$.
\end{lemma}

\begin{proof}
The lemma follows from the proof of Lemmas~\ref{lemma: 2s size} and~\ref{lemma: 4p complexity}:
in \eightap{}, $\tau(\calP)= \size{\calC}$,
and by the proof of
Lemma~\ref{lemma: 2s size},
$\size{\calC} =O\left( \frac{n\log n}{h} \right)$
w.p. at least $1-O(n^{-c+2})$.
Substituting this and $h = n^{4/11} \log^{10/11} n$
in Lemma~\ref{lemma: 4p complexity},
we get that the algorithm completes
in $O(n^{7/11} \log^{1/11} n + D)$ rounds
with the desired probability.
\end{proof}

\section{Lower Bounds}
\label{sec:lowerbounds}
In this section we prove lower bounds on the number
of rounds that are needed for constructing
spanners in the \cgst{} model.
All previous lower bounds for the distributed
construction of spanners~\cite{Pettie10}
use an indistinguishability argument:
while many edges should be omitted from the graph
in order to create a sparse spanner,
there are few edges that must not be omitted.
However, in order to distinguish these few edges
from the rest,
some nodes must learn a considerable part of the graph.
In a nutshell, the heart of the proof is that information must travel
a constant portion of the diameter $D$,
and thus the lower bound is $\Omega(D)$.

The lower bounds from \cite{Pettie10} apply also to the LOCAL model,
where the message sizes are unbounded.
Here,
we present the first lower bound
that is specific for the \cgst{} model.
As in previous lower bounds for the CONGEST model,
our proof uses a reduction from a communication complexity problem.
However, previous lower bounds used reductions
either from the equality problem~\cite{PelegR99}
or from set-disjointness, e.g.,%
~\cite{SarmaHKKNPPW12, DruckerKO13, FrischknechtHW12, HolzerP14, Censor-HillelGK14, GhaffariK13}.
These seem unsuitable for our purposes,
and hence we diverge from this approach and define
a new communication complexity problem
we call {\em partial complement}.
We bound the communication complexity
of this problem from below, using information theory.

We first prove a lower bound for the construction of a
$(+2)$-pairwise spanner.
Then, we generalize the bound for the construction of an $(\alpha,\beta)$-pairwise spanner,
for any $\alpha \geq 1, \beta \geq 0$.

\subsection{A Communication Complexity Problem}

Let $m,p$ be two positive integers so that $p \leq m/3$.
The \emph{partial complement} communication problem,
denoted \pc{}, is defined as follows:
\alice{} has a set $x \subseteq \set{1,\ldots, m}$
of size $\size{x} = p$,
and \bob{} has to output a set $y \subseteq \set{1,\ldots, m}$
of size $\size{y} = m/2$
so that $x \cap y = \emptyset$.
Note that the goal of this communication problem
is to compute a relation, not a function.
In this section we prove that
the randomized communication complexity
of the partial complement problem is high
(for formal definitions
in communication complexity see the textbook~\cite{KushilevitzN1996}).

\begin{theorem}
\label{thm:R>p}
If $\pi$ is a $(1/3)$-error randomized protocol
computing \pc{} then the length of $\pi$
is at least $p/100$.
\end{theorem}

The proof uses information theory.
The basic idea is to prove that any protocol that solves
\pc{} must contain $\Omega(p)$ bits
of information about the set $x$.
We now define the basic notions we use in the proof
(see the textbook~\cite{CoverT2006} for more background
and useful properties).

The \emph{entropy} of a random variable $X$ is defined as
$$\ent(X) = \sum_x \Pr[X=x] \log(1/ \Pr[X=x]).$$
It is well-known that entropy is maximal
for the uniform distribution:
if $X$ is an element chosen at random
from a set $T$ of size $t$,
then, by convexity,
\begin{equation} \label{eq: unifom entropy}
\ent[X]  =   \sum_{x\in T} \Pr[X=x] \log(1/\Pr[X=x])
\leq \log t.
\end{equation}
Shannon's coding theorem says that $\ent[X]$ is
a lower bound on the expected length of a
prefix-free encoding of $X$.
The entropy of a random variable $X$ conditioned on
a random variable $Y$ is defined as
$$\ent(X|Y) = \ent(X,Y) - \ent(Y).$$
The \emph{mutual information} between $X$ and $Y$ is
$$\I[X;Y] = \ent[X] - \ent[X\vert Y].$$
Intuitively, it measures the reduction in uncertainty about $X$
given the value of $Y$.

\begin{proof}[Proof of Theorem~\ref{thm:R>p}]
Let $\pi$ be a randomized protocol computing \pc{}
with $c$ bits of communication and error probability $1/3$.
Note that we can assume without loss of generality that only \alice{} speaks in $\pi$.
Below, we feed into $\pi$ a random set $X$,
so we may assume that $\pi$ is deterministic.
Denote by $\Pi = \Pi(x)$
the transcript of $\pi$ with input $x$.
Denote by $Y = Y(x)$ the subset of $\{1,\ldots,m\}$
of size $m/2$ that \bob{} outputs
when receiving message $\Pi(x)$.

The distribution on $X$ is defined as follows.
Let $I_1,\ldots,I_p$
be $p$ consecutive intervals
in $\set{1, \ldots, m}$, each of size $k = \lfloor m/p \rfloor$.
From each interval $I_i$,
choose one element $X_i$ uniformly at random and independently.
Set
$$X = \{X_1,\ldots,X_p\}.$$

First, we compute the entropy of $X$.
For each $i$, we have $\ent[X_i] = \log k$.
Thus, independence implies that
\[\ent[X] = \sum_{i=1}^p \ent[X_i]
        = p \log k.
\]

Second, we analyze
the entropy of $X$ conditioned on $Y$
and on the protocol outputting the correct answer.
Let $E$ be the indicator of the event
$$\{x \subset \{1,\ldots,m\} : |x|=p , \ Y(x) \cap x = \emptyset\}.$$
Note that since the protocol has an error probability bounded by $1/3$, it holds that $\Pr[E=1]\geq 2/3$.
Given $y$ of size $m/2$, split the intervals into two sets:
\[
S_1(y) = \set{i \in \{1,\ldots,p\} :  \size{I_i\cap y}\leq \frac{k}{10}} ,
\
S_2(y) = \set{i \in \{1,\ldots,p\} :  \size{I_i\cap y}> \frac{k}{10}}
.\]
Since $\size{S_1(y)} + \size{S_2(y)} = p$, it holds that
\begin{align*}
\frac{m}{2}
& = \size{y} \\
& \leq
\size{S_1(y)}\frac{k}{10} + \size{S_2(y)} k
+  k \\
& \leq
\frac{pk}{10} + \size{S_2(y)}\frac{9k}{10}
+ k,
\end{align*}
which by a simple calculation implies
$$\size{S_2(y)}\geq \frac{4}{9} p - \frac{10}{9}.$$
In addition, for each $y$ and
$i\in S_1(y)$,
using~\eqref{eq: unifom entropy}, we have
$$\ent[X_i \vert Y = y , E=1 ] \leq \log k.$$
For each $y$ and $i\in S_2(y)$,
conditioned on $E=1$ and $Y=y$,
the element $X_i$ is chosen
from $I_i \setminus y$,
whose size is at most
$k - \frac{k}{10} = \frac{9k}{10}$.
Thus, for each $y$ and $i\in S_2(y)$,
$$\ent[X_i \vert Y=y,E=1] \leq \log \frac{9k}{10}.$$
We can conclude this part as follows:\footnote{Note
the difference between conditioning on
an event and conditioning on a random variable.}
\begin{align*}
\ent[X\vert Y,E=1]&
\leq \sum_y \Pr[Y=y \vert E=1]
\left( \sum_{i\in S_1(y)} \ent[X_i \vert Y = y, E=1]
                +\sum_{i\in S_2(y)} \ent[X_i \vert Y =y , E=1]\right)
                \tag{sub-additivity}\\
 &\leq \sum_y \Pr[Y=y \vert E=1]
        \left((p - \size{S_2(y)}) \log k
               +   \size{S_2(y)} \log \frac{k}{10}\right) \\
&\leq  \sum_y \Pr[Y=y \vert E=1]      \left(   p \log k
                -   \size{S_2(y)}
                  \log \frac{10}{9} \right)\\
&\leq            p  \left( \log k
                -   \frac{4}{9} \log \frac{10}{9} \right) +2.
\end{align*}
Thus, using~\eqref{eq: unifom entropy},
 \begin{align*}
\ent[X\vert Y,E]
&=\Pr[E = 1] \cdot \ent[X\vert Y,E = 1] +
                \Pr[E = 0] \cdot \ent[X\vert Y,E = 0] \\
&\leq \frac{2}{3} \left( p  \left( \log k
                -   \frac{4}{9} \log \frac{10}{9} \right) +2 \right) +
                \frac{1}{3} \cdot p \log k \tag{$\Pr[E=1] \geq 2/3$} \\
& \leq 2+p \log k
                -   \frac{4p}{9} \log \frac{10}{9} .
\end{align*}

We now have
\begin{align*}
\I[X;Y]
& = \I[X;Y,E] - \I[E;X|Y]
\tag{chain rule} \\
& = \ent[X] - \ent[X|Y,E] - \I[E;X|Y] \\
& \geq \frac{4p}{9} \log \frac{10}{9}  - 3,
\end{align*}
where the last inequality holds
since $E$ contains at most one bit of information.
Therefore,
\begin{align*}
c &\geq \ent(\Pi) \tag{Shannon's coding theorem}
\\
    & \geq \I(X;\Pi) \\
    & \geq \I(X;Y) \tag{information processing inequality} \\
    & \geq \frac{4p}{9} \log \frac{10}{9}  - 3  \\
    & \geq p/100 \tag{we may assume $p > 100$};
\end{align*}
where $c$ is the length of $\pi$, completing the proof.
\end{proof}

\subsection{A Lower Bound for constructing a $(+2)$-Pairwise Spanner}
\label{sec: +2 lb}

In this section we prove the following lower bound.

\begin{theorem}
\label{thm: 2p lb}
There is a constant $c > 0$ so that the following holds.
Any distributed protocol for the \cgst{} model
with success probability at least $2/3$ which,
given a graph with $n$ nodes
and a set of $p \leq c n^{3/2}$ pairs of nodes,
outputs a $(+2)$-pairwise spanner with at most $c n^{3/2}$ edges,
must take $\Omega( \frac{p}{n\log n} )$
rounds to complete.
The lower bound holds
even for graphs with constant diameter.
\end{theorem}

The theorem implies a lower bound of
$\Omega\left(\sqrt{n}/\log n\right)$
on the number of rounds needed
for an algorithm in the \cgst{} model
to output a $(+2)$-pairwise spanner,
when $\size{\calP} = \Theta(n^{3/2})$.
For comparison, the time for constructing such a spanner
using \twop{} can (roughly) vary between $n^{3/4}$ and $n$,
depending on the structure of $\calP$.

The graph $G$ for which the lower bound is proved
is defined as follows.
Let $n$ be such that there is a finite projective plane with
$n/4$ points and $n/4$ lines.
Let $G'$ be the point-line incidence graph with $n/2$ nodes
(see e.g.~\cite[\S 4.5]{Matousek2002}).
The graph $G'$ has $\Theta(n^{3/2})$ edges, girth
$6$ and diameter $3$.\footnote{The \emph{girth} of a graph
is the length of the shortest simple cycle in it.}
We denote the nodes of $G'$ by
$V_B= \{v_1',\ldots, v_{n/2}'\}$.
The graph $G$
consists of $G'$,
an additional $n/2$ nodes denoted
$V_A = \{v_1,\ldots, v_{n/2}\}$,
and an additional $n/2$ edges of the form $(v_i, v_i')$.

In the pairwise spanner we construct,
we wish to approximately preserve distances between pairs of nodes in $V_A$,
i.e.\ $\calP\subseteq V_A\times V_A$.
The main observation is that,
since the girth of $G'$ is $6$,
if $e' = \{v'_i,v'_j\}$
is an edge of $G'$ then the following holds.
If $(v_i, v_j)\in \calP$
then any $(+2)$-pairwise spanner
must contain the edge $e'$,
as otherwise the distance is stretched from $3$ to $7$,
which exceeds the required $+2$ stretch.
On the other hand, if $(v_i, v_j)\notin \calP$ then the
edge $e'$ can be safely omitted from the spanner.

\begin{proof}[Proof of Theorem~\ref{thm: 2p lb}]
Fix a distributed protocol $\sigma$ for constructing
a $(+2)$-pairwise spanner with at most $m/2$ edges.
Let $G$ be the graph described
above,
and denote the edges of $G'$ by $e_1,\ldots, e_m$.

We describe a reduction
from \pc{} to $\sigma$.
Assume \alice{} has a set
$x \subseteq \set{1,\ldots, m}$ of size $p$,
and \bob{} has to output a set
$y \subseteq\set{1,\ldots, m}$ of size $m/2$
satisfying $x\cap y = \emptyset$.
\alice{} and \bob{} simulate $\sigma$
on the graph $G$ with the set of pairs
$$\calP= \set{(v_i,v_j) : \exists k\in x \
e_k=\{v_i', v_j'\}}.$$
That is, a pair $(v_i,v_j)$ is in $\calP$ if the corresponding
pair $(v_i',v_j')$ is an edge $e_k$ whose index $k$
is in $x$.
\alice{} simulates the nodes in $V_A$,
and \bob{} simulates the nodes $V_B$ and the edges among them.
To simulate communication on edges of the form $(v_i, v_i')$,
\alice{} and \bob{} communicate.
Note that $\calP$ contains only pairs of nodes
that are simulated by \alice{}.

The spanner constructed is a subgraph $H$ of $G$
with at most $m/2$ edges,
satisfying $\delta_H(v_i,v_j)\leq \delta_G(v_i,v_j)+2$
for all $(v_i, v_j)\in \calP$.
For each such pair, by definition of $\calP$, we have
$\delta_G (v_i, v_j) = 3$,
which implies
 $\delta_H (v_i, v_j) \leq 5$
 and $\delta_H (v_i', v_j') \leq 3$.
The fact that $G'$ has girth $6$ implies
that the edge $\{v_i', v_j'\}$ must be in $H$.
Let $$y= \set{k : e_k\in E_G\setminus E_H}.$$
The spanner size implies $\size{y}\geq m/2$,
while the above discussion implies $x\cap y = \emptyset$.
Thus, \bob{} can output a subset of $y$ of size $m/2$,
solving the communication complexity problem.

By the communication complexity lower bound,
\alice{} and \bob{} must communicate $\Omega(p)$ bits
during the simulation.
The number of edges they simulate together is $n/2$,
and $O(\log n)$ bits are sent over each edge
at each round.
Thus, the protocol must take
$\Omega\left( \frac{\size{\calP}}{n\log n} \right)$
rounds to complete.
\end{proof}

\subsection{Generalization: A Lower bound for constructing an $(\alpha,\beta)$-Pairwise Spanner}

Recall that an $\left(\alpha,\beta\right)$-pairwise spanner for a graph $G$
and a set $\calP$ of pairs of nodes
is a subgraph $H$ of $G$ satisfying
$\delta_H(u,v)\leq \alpha \delta_G(u,v)+\beta$
for every $(u,v)\in\calP$.

To obtain our lower bound for any $\alpha \geq 1, \beta \geq 0$,
we first study the tradeoff between the girth
and number of edges in a graph.
The most relevant claim for this question is Erd\H{o}s' girth conjecture:

\begin{conjecture}[Erd\H{o}s' Girth Conjecture~\cite{Erdos64}]
For every $g$ there is a constant $c$
such that there exists a graph on $n$ nodes
with girth $g$ and $cn^{1+\frac{1}{\ceil{g/2}-1}}$ edges.
\end{conjecture}

For example,
for $g=3$, the complete graph on $n$ nodes has roughly $n^2/2$ edges,
and for $g=4$ the full bipartite graph
has $n^2/4$ edges.
For $g=5$ and $g=6$ there exist a graph with $n^{3/2}$ edges,
which we used in the last section.

The conjecture is known to be true
for a few values of $g$,
while for the other values there are constructions
with slightly less edges:
\begin{theorem}[see, e.g.{~\cite[\S 15.3]{Matousek2002}}]
\label{thm: girth edges tradeoff}
For every $g\geq 3$ and $n\geq 2$
there is a graph on $n$ vertices with girth $g$ and
$\Omega\left(n^{1+\frac{4}{3g-10}} \right)$ edges.
For $g\in \set{3,4,5,6,9,10}$ there is a constant $c$
such that for every $n\geq 2$
there is a graph on $n$ vertices with girth $g$
and $cn^{1+\frac{1}{\ceil{g/2}-1}}$ edges.
\end{theorem}

Theorem~\ref{thm: 2p lb} and its proof extend to
$(\alpha,\beta)$-pairwise spanners for any $\alpha\geq 1$
and $\beta\geq 0$,
with the appropriate choice of $G'$.
The only thing left is to make sure the diameter is constant
even if the diameter of $G'$ is not.
Bounding the diameter also allows us to derive a lower bound
conditioned on Erd\H{o}s' girth conjecture,
a conjecture which claims nothing about the diameter
of the graph.

\begin{theorem}
\label{thm: beta-p lb}
Let $\alpha\geq 1, \beta\geq 0$ be constants, and $g=3\alpha+\beta$.
There is a constant $c = c(g)$ so that the following holds.
Any distributed protocol for the \cgst{} model
with success probability at least $2/3$ which,
given a graph with $n$ nodes
and a set of $p \leq c n^{1+\frac{4}{3g-10}}$
pairs of nodes,
outputs an $(\alpha,\beta)$-pairwise spanner with at most $c n^{1+\frac{4}{3g-10}}$ edges,
must take $\Omega\left( \frac{p}{n\log n} \right)$
rounds to complete.
The lower bound holds even for graphs with diameter
$O(g)$.

For $g\in \set{3,4,5,6,9,10}$,
and for any constant $g$ if Erd\H{o}s' girth conjecture is true,
the bound on $p$ and on the number of edges
can be replaced by $cn^{1+\frac{1}{\ceil{g/2}-1}}$.
\end{theorem}

The theorem implies a lower bound of
$\Omega\left(n/\log(n)\right)$ rounds
for any algorithm in the \cgst{} model
which outputs pairwise preserver
with $o(n^2)$ edges,
when $\size{\calP} = \Theta(n^2)$.
The trivial algorithm that builds a BFS tree
from any node appearing in $\calP$
runs in $O(n)$ rounds and returns
a spanner with $O(n^2)$ edges.
Hence, we cannot expect to asymptotically improve
upon the running time and upon the number of edges simultaneously.

For $(+4)$-pairwise spanners,
the theorem implies a lower bound of
$\Omega(n^{4/11}/\log n)$ rounds
for $\size{\calP} = \Theta(n^{15/11})$.
\fourp{} constructs such a spanner
in roughly $O(n^{107/77}\log^{6/7}n)$ rounds.
Assuming Erd\H{o}s' conjecture,
we can choose $\calP$ satisfying
$\size{\calP} = \Theta(n^{4/3})$,
and get a lower bound of $\Omega(n^{1/3}/\log n)$
rounds
and an algorithm running in
$O(n^{29/21}\log^{6/7}n)$ rounds.

The graph $G$ for which the lower bound is proved is defined
similarly to the graph in the proof of Theorem~\ref{thm: 2p lb},
with an extra construction that ensures that the diameter of $G$
is constant.
Let $G'$ be a graph on $n'$ nodes with girth $g=3\alpha+\beta$
and $m$ edges,
where $m=m(n',g)$ is the maximal possible number of edges
for these $n'$ and $g$.
Let $v_1',\ldots, v_{n'}'$ be the nodes of $G'$.
Add to $G'$ a new node $u$
and connect each node $v_i'$ to $u$ by a disjoint path of
$\floor{g/2}$ new nodes.
This increases the number of nodes in $G'$ and the number
of edges by a multiplicative $g$ factor,
does not decrease the girth,
and ensures the diameter is $O(g)$.

The lower bound graph $G$ consists of
$G'$, the node $u$,
the nodes and paths connecting each $v_i'$ to $u$,
and for each $v_i'$ another node $v_i$
connected to it by an edge $(v_i, v_i')$.
Let $V_A = \set{v_1,\ldots, v_{n'}}$ and $V_B$
be the set of all other nodes.

The pairwise spanner we construct approximately preserves
distances between pairs of nodes in $V_A$,
i.e.\ $\calP\subseteq V_A\times V_A$.
Since the girth of $G'$ is $3\alpha + \beta$,
if $e' = \{v'_i,v'_j\}$ is an edge in $G'$ and
$(v_i, v_j)\in \calP$,
then any $(\alpha,\beta)$-pairwise spanner must
contain the edge $e'$:
in $G$, the distance between $v_i$ and $v_j$ is $3$;
and if there is a path in $H$ connecting $v_i$ and $v_j$
with length at most $3\alpha+\beta$,
then this path connects $v_i'$ and $v_j'$ in $G'$
with $3\alpha+\beta -2$ edges,
which together with $e'$ closes a cycle
of length $3\alpha+\beta -1 <g$ in $G'$.
On the other hand, if $(v_i, v_j)\notin \calP$ then the
edge $e'$ can be omitted from the spanner.

\begin{proof}[Proof of Theorem~\ref{thm: beta-p lb}]
Fix a distributed protocol $\sigma$ for constructing
an $(\alpha,\beta)$-pairwise spanner.
Let $G$ be the graph described above,
and denote the edges of $G'$ by $e_1,\ldots, e_m$.
Theorem~\ref{thm: girth edges tradeoff} gives the bound
$m \geq 2c n^{1+\frac{4}{3g-10}}$ for some constant $c>0$.

We use the same reduction from \pc{} to $\sigma$:
\alice{} has a set $x \subseteq \set{1,\ldots, m}$
of size $p$,
and \bob{} has to output a set
$y \subseteq\set{1,\ldots, m}$ of size $m/2$
satisfying $x\cap y = \emptyset$.
They simulate $\sigma$ on the graph $G$ and the set of pairs
$\calP= \set{(v_i,v_j) : \exists k\in x \
e_k=\{v_i', v_j'\}}$;
\alice{} simulates the nodes in $V_A$
and \bob{} simulates the nodes $V_B$ and the edges among them.
To simulate communication on edges of the form $(v_i, v_i')$,
\alice{} and \bob{} communicate.
The bound on $p$ in the theorem statement comes from the fact that $\size{\calP}\leq m/2$.

The spanner constructed is a subgraph $H$ of $G$
with at most $m/2$ edges,
satisfying $\delta_H(v_i,v_j)\leq \alpha\cdot \delta_G(v_i,v_j)+\beta$
for all $(v_i, v_j)\in \calP$.
For each such pair, the edge $\{v_i', v_j'\}$ must be in $H$,
as explained above.
Let $y= \set{k : e_k\in E_G\setminus E_H}.$
The spanner size implies $\size{y}\geq m/2$,
while the above discussion implies $x\cap y = \emptyset$.
Thus, \bob{} can output a subset of $y$ of size $m/2$,
solving the communication complexity problem.

By the communication complexity lower bound,
\alice{} and \bob{} must communicate $\Omega(p)$ bits
during the simulation.
The number of edges they simulate together is at most $n$,
and $O(\log n)$ bits are transferred over each edge
at each round.
Thus, the protocol must take
$\Omega\left( \frac{\size{\calP}}{n\log n} \right)$
rounds to complete.
\end{proof}

\section{Discussion}
\label{sec:discussion}
This paper presents various algorithms for computing sparse purely additive spanners in the CONGEST model.
Our algorithms exhibit tradeoffs between the running time and the sparsity of the constructed spanners. By choosing different values for the parameter $h$,
one can obtain a spanner with the same stretch in a smaller number of rounds but at the expense of increasing the density.
This tradeoff is an important direction for future work.

Our lower bound uses a new communication complexity problem, and leverages the distributed nature of the system by using the fact that each node initially only knows the pairs in $\calP$ to which it belongs\footnote{In fact, our lower bound holds even if all nodes in pairs in $\calP$ know all of $\calP$.}. That is,
the topology of the graph used for the lower bound reduction is known completely to both \alice{} and \bob{}, regardless of their inputs to the \pc{} instance, while the uncertainty about the identity of the pairs in $\calP$ is what makes the problem hard.
While it might be unnatural to assume that other nodes know about these pairs, it is theoretically interesting to ask whether one can design faster distributed constructions given this information.

Finally, we believe that our new lower bound technique can be useful for proving additional lower bounds in the CONGEST model, as it diverges from reducing to the set-disjointness problem.

\paragraph{Acknowledgements}
We thank Merav Parter for a helpful discussion on the lower bound,
and the anonymous referees of DISC 2016 for valuable comments.

\let\oldbibliography\thebibliography
\renewcommand{\thebibliography}[1]{%
  \oldbibliography{#1}%
  \setlength{\itemsep}{0pt}%
}

\bibliographystyle{alpha}
\bibliography{Bibliography}

\newcommand{\etalchar}[1]{$^{#1}$}
\begin{thebibliography}{DSHK{\etalchar{+}}12}

\bibitem[AB16a]{AbboudB16-2}
Amir Abboud and Greg Bodwin.
\newblock The 4/3 additive spanner exponent is tight.
\newblock In {\em ACM SIGACT Symposium on Theory of Computing, {STOC}}, 2016.

\bibitem[AB16b]{AbboudB16-1}
Amir Abboud and Greg Bodwin.
\newblock Error amplification for pairwise spanner lower bounds.
\newblock In {\em Proceedings of the Twenty-Seventh Annual {ACM-SIAM} Symposium
  on Discrete Algorithms, {SODA}}, pages 841--854, 2016.

\bibitem[ACIM99]{ACIM99}
Donald Aingworth, Chandra Chekuri, Piotr Indyk, and Rajeev Motwani.
\newblock Fast estimation of diameter and shortest paths (without matrix
  multiplication).
\newblock {\em {SIAM} J. Comput.}, 28(4):1167--1181, 1999.

\bibitem[ADD{\etalchar{+}}93]{ADDJS93}
Ingo Alth{\"{o}}fer, Gautam Das, David~P. Dobkin, Deborah Joseph, and
  Jos{\'{e}} Soares.
\newblock On sparse spanners of weighted graphs.
\newblock {\em Discrete {\&} Computational Geometry}, 9:81--100, 1993.

\bibitem[Bas08]{Baswana08}
Surender Baswana.
\newblock Streaming algorithm for graph spanners - single pass and constant
  processing time per edge.
\newblock {\em Inf. Process. Lett.}, 106(3):110--114, 2008.

\bibitem[BCE05]{BCE03}
B{\'{e}}la Bollob{\'{a}}s, Don Coppersmith, and Michael Elkin.
\newblock Sparse distance preservers and additive spanners.
\newblock {\em {SIAM} J. Discrete Math.}, 19(4):1029--1055, 2005.

\bibitem[BKMP10]{BaswanaKMP10}
Surender Baswana, Telikepalli Kavitha, Kurt Mehlhorn, and Seth Pettie.
\newblock Additive spanners and (alpha, beta)-spanners.
\newblock {\em {ACM} Trans. Algorithms}, 7(1):5, 2010.

\bibitem[BKS12]{BaswanaKS12}
Surender Baswana, Sumeet Khurana, and Soumojit Sarkar.
\newblock Fully dynamic randomized algorithms for graph spanners.
\newblock {\em {ACM} Trans. Algorithms}, 8(4):35, 2012.

\bibitem[BS07]{BS03}
Surender Baswana and Sandeep Sen.
\newblock A simple and linear time randomized algorithm for computing sparse
  spanners in weighted graphs.
\newblock {\em Random Struct. Algorithms}, 30(4):532--563, 2007.

\bibitem[BS08]{BaswanaS08}
Surender Baswana and Soumojit Sarkar.
\newblock Fully dynamic algorithm for graph spanners with poly-logarithmic
  update time.
\newblock In {\em Proceedings of the Nineteenth Annual {ACM-SIAM} Symposium on
  Discrete Algorithms, {SODA}}, pages 1125--1134, 2008.

\bibitem[BW16]{BW16}
Greg Bodwin and Virginia~Vassilevska Williams.
\newblock Better distance preservers and additive spanners.
\newblock In {\em Proceedings of the Twenty-Seventh Annual {ACM-SIAM} Symposium
  on Discrete Algorithms, {SODA}}, pages 855--872, 2016.

\bibitem[CE06]{CE05}
Don Coppersmith and Michael Elkin.
\newblock Sparse sourcewise and pairwise distance preservers.
\newblock {\em {SIAM} J. Discrete Math.}, 20(2):463--501, 2006.

\bibitem[CGK13]{CGK13}
Marek Cygan, Fabrizio Grandoni, and Telikepalli Kavitha.
\newblock On pairwise spanners.
\newblock In {\em 30th International Symposium on Theoretical Aspects of
  Computer Science, {STACS}}, pages 209--220, 2013.

\bibitem[CGK14]{Censor-HillelGK14}
Keren Censor{-}Hillel, Mohsen Ghaffari, and Fabian Kuhn.
\newblock Distributed connectivity decomposition.
\newblock In {\em {ACM} Symposium on Principles of Distributed Computing,
  {PODC}}, pages 156--165, 2014.

\bibitem[Che13a]{Chechik13a}
Shiri Chechik.
\newblock Compact routing schemes with improved stretch.
\newblock In {\em {ACM} Symposium on Principles of Distributed Computing,
  {PODC}}, pages 33--41, 2013.

\bibitem[Che13b]{Chechik13}
Shiri Chechik.
\newblock New additive spanners.
\newblock In {\em Proceedings of the Twenty-Fourth Annual {ACM-SIAM} Symposium
  on Discrete Algorithms, {SODA}}, pages 498--512, 2013.

\bibitem[CHKM12]{Censor-HillelHKM12}
Keren Censor{-}Hillel, Bernhard Haeupler, Jonathan~A. Kelner, and Petar
  Maymounkov.
\newblock Global computation in a poorly connected world: fast rumor spreading
  with no dependence on conductance.
\newblock In {\em Proceedings of the 44th Symposium on Theory of Computing
  Conference, {STOC}}, pages 961--970, 2012.

\bibitem[CT06]{CoverT2006}
Thomas~M. Cover and Joy~A. Thomas.
\newblock {\em Elements of Information Theory (Wiley Series in
  Telecommunications and Signal Processing)}.
\newblock Wiley-Interscience, 2006.

\bibitem[DG08]{DerbelG08}
Bilel Derbel and Cyril Gavoille.
\newblock Fast deterministic distributed algorithms for sparse spanners.
\newblock {\em Theor. Comput. Sci.}, 399(1-2):83--100, 2008.

\bibitem[DGP07]{DerbelGP07}
Bilel Derbel, Cyril Gavoille, and David Peleg.
\newblock Deterministic distributed construction of linear stretch spanners in
  polylogarithmic time.
\newblock In {\em Distributed Computing, 21st International Symposium, {DISC}},
  pages 179--192, 2007.

\bibitem[DGPV08]{DerbelGPV08}
Bilel Derbel, Cyril Gavoille, David Peleg, and Laurent Viennot.
\newblock On the locality of distributed sparse spanner construction.
\newblock In {\em Proceedings of the 27th Annual {ACM} Symposium on Principles
  of Distributed Computing, {PODC}}, pages 273--282, 2008.

\bibitem[DGPV09]{DerbelGPV09}
Bilel Derbel, Cyril Gavoille, David Peleg, and Laurent Viennot.
\newblock Local computation of nearly additive spanners.
\newblock In {\em Distributed Computing, 23rd International Symposium, {DISC}},
  pages 176--190, 2009.

\bibitem[DHZ00]{DHZ97}
Dorit Dor, Shay Halperin, and Uri Zwick.
\newblock All-pairs almost shortest paths.
\newblock {\em {SIAM} J. Comput.}, 29(5):1740--1759, 2000.

\bibitem[DKO14]{DruckerKO13}
Andrew Drucker, Fabian Kuhn, and Rotem Oshman.
\newblock On the power of the congested clique model.
\newblock In {\em {ACM} Symposium on Principles of Distributed Computing,
  {PODC}}, pages 367--376, 2014.

\bibitem[DMP{\etalchar{+}}05]{DubhashiMPRS05}
Devdatt~P. Dubhashi, Alessandro Mei, Alessandro Panconesi, Jaikumar
  Radhakrishnan, and Aravind Srinivasan.
\newblock Fast distributed algorithms for (weakly) connected dominating sets
  and linear-size skeletons.
\newblock {\em J. Comput. Syst. Sci.}, 71(4):467--479, 2005.

\bibitem[DSHK{\etalchar{+}}12]{SarmaHKKNPPW12}
Atish Das~Sarma, Stephan Holzer, Liah Kor, Amos Korman, Danupon Nanongkai,
  Gopal Pandurangan, David Peleg, and Roger Wattenhofer.
\newblock Distributed verification and hardness of distributed approximation.
\newblock {\em {SIAM} J. Comput.}, 41(5):1235--1265, 2012.

\bibitem[Elk05]{E05}
Michael Elkin.
\newblock Computing almost shortest paths.
\newblock {\em {ACM} Trans. Algorithms}, 1(2):283--323, 2005.

\bibitem[Elk07a]{Elkin07}
Michael Elkin.
\newblock A near-optimal distributed fully dynamic algorithm for maintaining
  sparse spanners.
\newblock In {\em Proceedings of the Twenty-Sixth Annual {ACM} Symposium on
  Principles of Distributed Computing, {PODC}}, pages 185--194, 2007.

\bibitem[Elk07b]{Elkin07a}
Michael Elkin.
\newblock Streaming and fully dynamic centralized algorithms for constructing
  and maintaining sparse spanners.
\newblock In {\em Automata, Languages and Programming, 34th International
  Colloquium, {ICALP}}, pages 716--727, 2007.

\bibitem[EP04]{EP04}
Michael Elkin and David Peleg.
\newblock (1+epsilon, beta)-spanner constructions for general graphs.
\newblock {\em {SIAM} J. Comput.}, 33(3):608--631, 2004.

\bibitem[Erd64]{Erdos64}
P.~Erd\H{o}s.
\newblock Extremal problems in graph theory.
\newblock In {\em Theory of graphs and its applications: proceedings of the
  symposium held in Smolenice in June 1963}, pages 29--36. Pub. House of the
  Czechoslovak Academy of Sciences, 1964.

\bibitem[EZ06]{ElkinZ06}
Michael Elkin and Jian Zhang.
\newblock Efficient algorithms for constructing (1+epsilon, beta)-spanners in
  the distributed and streaming models.
\newblock {\em Distributed Computing}, 18(5):375--385, 2006.

\bibitem[FHW12]{FrischknechtHW12}
Silvio Frischknecht, Stephan Holzer, and Roger Wattenhofer.
\newblock Networks cannot compute their diameter in sublinear time.
\newblock In {\em Proceedings of the Twenty-Third Annual {ACM-SIAM} Symposium
  on Discrete Algorithms, {SODA}}, pages 1150--1162, 2012.

\bibitem[GK13]{GhaffariK13}
Mohsen Ghaffari and Fabian Kuhn.
\newblock Distributed minimum cut approximation.
\newblock In {\em Distributed Computing - 27th International Symposium,
  {DISC}}, pages 1--15, 2013.

\bibitem[HP14]{HolzerP14}
Stephan Holzer and Nathan Pinsker.
\newblock Approximation of distances and shortest paths in the broadcast
  congest clique.
\newblock {\em CoRR}, abs/1412.3445, 2014.

\bibitem[HW12]{HolzerW12}
Stephan Holzer and Roger Wattenhofer.
\newblock Optimal distributed all pairs shortest paths and applications.
\newblock In {\em {ACM} Symposium on Principles of Distributed Computing,
  {PODC}}, pages 355--364, 2012.

\bibitem[Kav15]{Kavitha15}
Telikepalli Kavitha.
\newblock New pairwise spanners.
\newblock In {\em 32nd International Symposium on Theoretical Aspects of
  Computer Science, {STACS}}, pages 513--526, 2015.

\bibitem[KN97]{KushilevitzN1996}
Eyal Kushilevitz and Noam Nisan.
\newblock {\em Communication Complexity}.
\newblock Cambridge University Press, New York, NY, USA, 1997.

\bibitem[Knu14]{knudsen14}
Mathias B{\ae}k~Tejs Knudsen.
\newblock Additive spanners: {A} simple construction.
\newblock In {\em 14th Scandinavian Symposium and Workshops on Algorithm
  Theory, {SWAT}}, pages 277--281, 2014.

\bibitem[KV13]{KavithaV13}
Telikepalli Kavitha and Nithin~M. Varma.
\newblock Small stretch pairwise spanners.
\newblock In {\em 40th International Colloquium on Automata, Languages, and
  Programming, {ICALP}}, pages 601--612, 2013.

\bibitem[LP13]{LenzenP13}
Christoph Lenzen and David Peleg.
\newblock Efficient distributed source detection with limited bandwidth.
\newblock In {\em {ACM} Symposium on Principles of Distributed Computing,
  {PODC}}, pages 375--382, 2013.

\bibitem[Mat02]{Matousek2002}
Jiri Matousek.
\newblock {\em Lectures on Discrete Geometry}.
\newblock Springer-Verlag New York, Inc., Secaucus, NJ, USA, 2002.

\bibitem[Par14]{Par14}
Merav Parter.
\newblock Bypassing erd{\H{o}}s' girth conjecture: Hybrid stretch and
  sourcewise spanners.
\newblock In {\em 41st International Colloquium on Automata, Languages, and
  Programming, {ICALP}}, pages 608--619, 2014.

\bibitem[Pel00]{Peleg_2000}
David Peleg.
\newblock {\em Distributed Computing: A Locality-Sensitive Approach}.
\newblock Monographs on Discrete Mathematics and Applications. Society for
  Industrial and Applied Mathematics, 2000.

\bibitem[Pet09]{Pettie09}
Seth Pettie.
\newblock Low distortion spanners.
\newblock {\em {ACM} Trans. Algorithms}, 6(1), 2009.

\bibitem[Pet10]{Pettie10}
Seth Pettie.
\newblock Distributed algorithms for ultrasparse spanners and linear size
  skeletons.
\newblock {\em Distributed Computing}, 22(3):147--166, 2010.

\bibitem[PR99]{PelegR99}
David Peleg and Vitaly Rubinovich.
\newblock A near-tight lower bound on the time complexity of distributed {MST}
  construction.
\newblock In {\em 40th Annual Symposium on Foundations of Computer Science,
  {FOCS}}, pages 253--261, 1999.

\bibitem[PS89]{PS89}
David Peleg and Alejandro~A. Sch{\"{a}}ffer.
\newblock Graph spanners.
\newblock {\em Journal of Graph Theory}, 13(1):99--116, 1989.

\bibitem[PU89a]{PelegU89a}
David Peleg and Jeffrey~D. Ullman.
\newblock An optimal synchronizer for the hypercube.
\newblock {\em {SIAM} J. Comput.}, 18(4):740--747, 1989.

\bibitem[PU89b]{PelegU89b}
David Peleg and Eli Upfal.
\newblock A trade-off between space and efficiency for routing tables.
\newblock {\em J. {ACM}}, 36(3):510--530, 1989.

\bibitem[RTZ05]{RMZ05}
Liam Roditty, Mikkel Thorup, and Uri Zwick.
\newblock Deterministic constructions of approximate distance oracles and
  spanners.
\newblock In {\em 32nd International Colloquium on Automata, Languages and
  Programming, {ICALP}}, pages 261--272, 2005.

\bibitem[RZ11]{RZ04}
Liam Roditty and Uri Zwick.
\newblock On dynamic shortest paths problems.
\newblock {\em Algorithmica}, 61(2):389--401, 2011.

\bibitem[TZ01]{ThorupZ01}
Mikkel Thorup and Uri Zwick.
\newblock Compact routing schemes.
\newblock In {\em {SPAA}}, pages 1--10, 2001.

\bibitem[TZ05]{TZ01}
M.~Thorup and U.~Zwick.
\newblock Approximate distance oracles.
\newblock {\em J. ACM}, 52(1), 2005.

\bibitem[TZ06]{TZ06}
Mikkel Thorup and Uri Zwick.
\newblock Spanners and emulators with sublinear distance errors.
\newblock In {\em Proceedings of the 17th Annual {ACM-SIAM} Symposium on
  Discrete Algorithms, {SODA}}, pages 802--809, 2006.

\bibitem[Woo10]{Woodruff_2010}
David~P. Woodruff.
\newblock Additive spanners in nearly quadratic time.
\newblock In {\em 37th International Colloquium on Automata, Languages and
  Programming, {ICALP}}, pages 463--474, 2010.

\end{thebibliography}

\end{document}